\documentclass[11pt]{article}
\RequirePackage{amsthm,amsmath,amsfonts,amssymb}
\usepackage[authoryear]{natbib}
\usepackage{amsfonts}
\usepackage{bbm}
\usepackage{mathrsfs}
\usepackage{subfigure}
\usepackage[pdftex]{graphicx}
\usepackage{titlesec}
\usepackage{lipsum}
\usepackage{enumerate}
\usepackage{amssymb}
\usepackage{amsmath,amssymb,amsfonts,amsthm}
\usepackage{bm}
\usepackage{here}
\usepackage{multicol}
\usepackage{epsf}
\usepackage{authblk}
\makeatother
\begin{document}
\newcommand{\field}[1]{\mathbb{#1}} 
\newcommand{\ind}{\mathbbm{1}}
\newcommand{\C}{\mathbb{C}}
\newcommand{\D}{\,\mathscr{D}}
\newcommand{\E}{\,\mathrm{E}}
\newcommand{\Prob}{\,\mathrm{P}}
\newcommand{\F}{\,\mathscr{F}}
\newcommand{\I}{\,\mathrm{i}}
\newcommand{\N}{\field{N}}
\newcommand{\T}{\,\mathrm{T}}
\newcommand{\Ls}{\,\mathscr{L}}
\newcommand{\ve}{\,\mathrm{vec}}
\newcommand{\var}{\,\mathrm{var}}
\newcommand{\cov}{\,\mathrm{Cov}}
\newcommand{\vech}{\,\mathrm{vech}}
\newcommand*\dif{\mathop{}\!\mathrm{d}}
\newcommand{\difs}{\mathrm{d}}
\newcommand\mi{\mathrm{i}}
\newcommand\me{\mathrm{e}}
\newcommand{\R}{\field{R}}

\newcommand{\es}{\hat{f}_n}
\newcommand{\ess}{\hat{f}_{\flr{ns}}}
\newcommand{\p}[1]{\frac{\partial}{\partial#1}}
\newcommand{\pp}[1]{\frac{\partial^2}{\partial#1\partial#1^{\top}}}
\newcommand{\para}{\bm{\theta}}
\providecommand{\bv}{\mathbb{V}}
\providecommand{\bu}{\mathbb{U}}
\providecommand{\bt}{\mathbb{T}}
\newcommand{\flr}[1]{\lfloor#1\rfloor}
\newcommand{\ba}{B_m}
\newcommand{\bxi}{\bar{\xi}_n}
\newcommand{\sgn}{{\rm sgn \,}}
\newcommand{\rint}{\int^{\infty}_{-\infty}}

\newcommand{\dr}{\mathrm{d}}
\newcommand{\red}[1]{\textcolor{red}{#1}}

\newcommand{\skakko}[1]{\left(#1\right)}
\newcommand{\mkakko}[1]{\left\{#1\right\}}
\newcommand{\lkakko}[1]{\left[#1\right]}

\newcommand{\Z}{\field{Z}}
\newcommand{\Zo}{\field{Z}_0}

\newcommand{\abs}[1]{\lvert#1\rvert}
\newcommand{\ct}[1]{\langle#1\rangle}
\newcommand{\inp}[2]{\langle#1,#2\rangle}
\newcommand{\norm}[1]{\lVert#1 \rVert}
\newcommand{\Bnorm}[1]{\Bigl\lVert#1\Bigr  \rVert}
\newcommand{\Babs}[1]{\Bigl \lvert#1\Bigr \rvert} 
\newcommand{\ep}{\epsilon} 
\newcommand{\sumn}[1][i]{\sum_{#1 = 1}^T}
\newcommand{\tsum}[2][i]{\sum_{#1 = -#2}^{#2}}
\providecommand{\abs}[1]{\lvert#1\rvert}
\providecommand{\Babs}[1]{\Bigl \lvert#1\Bigr \rvert} 

\newcommand{\uint}{\int^{1}_{0}}
\newcommand{\freqint}{\int^{\pi}_{-\pi}}
\newcommand{\li}[1]{\mathfrak{L}(S_{#1})}

\newcommand{\cum}{{\rm cum}}

\newcommand{\xt}{\bm{X}_{t, T}}
\newcommand{\yt}{\bm{Y}_{t, T}}
\newcommand{\zt}{\bm{Z}_{t, T}}
\newcommand{\gcu}{{\rm GC}^{2 \to 1}(u)}

\newcommand{\btheta}{\bm{\bm{\theta}}}
\newcommand{\bbeta}{\bm{\eta}}
\newcommand{\bzeta}{\bm{\zeta}}
\newcommand{\bzero}{\bm{0}}
\newcommand{\bI}{\bm{I}}
\newcommand{\bd}{\bm{d}}
\newcommand{\bx}[1]{\bm{X}_{#1, T}}
\newcommand{\be}{\bm{\ep}}
\newcommand{\bp}{\bm{\phi}}

\newcommand{\act}{A_T^{\circ}}
\newcommand{\ac}{A^{\circ}}

\newcommand{\dlim}{\xrightarrow{d}}
\newcommand{\plim}{\rightarrow_{P}}

\newcommand{\ls}{\mathcal{S}}
\newcommand{\cs}{\mathcal{C}}

\providecommand{\ttr}[1]{\textcolor{red}{ #1}}
\providecommand{\ttb}[1]{\textcolor{blue}{ #1}}
\providecommand{\ttg}[1]{\textcolor{green}{ #1}}
\providecommand{\tty}[1]{\textcolor{yellow}{ #1}}
\providecommand{\tto}[1]{\textcolor{orange}{ #1}}
\providecommand{\ttp}[1]{\textcolor{purple}{ #1}}

\newcommand{\sign}{\mathop{\rm sign}}
\newcommand{\conv}{\mathop{\rm conv}}
\newcommand{\argmax}{\mathop{\rm arg~max}\limits}
\newcommand{\argmin}{\mathop{\rm arg~min}\limits}
\newcommand{\argsup}{\mathop{\rm arg~sup}\limits}
\newcommand{\arginf}{\mathop{\rm arg~inf}\limits}
\newcommand{\diag}{\mathop{\rm diag}}
\newcommand{\minimize}{\mathop{\rm minimize}\limits}
\newcommand{\maximize}{\mathop{\rm maximize}\limits}
\newcommand{\tr}{\mathop{\rm tr}}
\newcommand{\Cum}{\mathop{\rm Cum}\nolimits}
\newcommand{\Var}{\mathop{\rm Var}\nolimits}
\newcommand{\Cov}{\mathop{\rm Cov}\nolimits}

\numberwithin{equation}{section}
\theoremstyle{plain}
\newtheorem{thm}{Theorem}[section]

\newtheorem{lem}[thm]{Lemma}
\newtheorem{prop}[thm]{Proposition}
\theoremstyle{definition}
\newtheorem{defi}[thm]{Definition}
\newtheorem{assumption}[thm]{Assumption}
\newtheorem{cor}[thm]{Corollary}
\newtheorem{rem}[thm]{Remark}
\newtheorem{eg}[thm]{Example}

\title{Identification and estimation 
of structural vector autoregressive models 
via LU decomposition}
\author[1]{Masato Shimokawa}
\author[2]{Kou Fujimori}
\affil[1, 2]{Faculty of Economics and Law,
Shinshu University.}
\date{}
\maketitle
\begin{abstract}
Structural vector autoregressive (SVAR) 
models are 
widely used to analyze the simultaneous relationships 
between multiple time-dependent data.
Various statistical inference methods 
have been studied to 
overcome the identification problems of SVAR models.
However, most of these methods impose 
strong assumptions for innovation processes
such as the uncorrelation of components.
In this study, we relax the assumptions for 
innovation processes and propose an identification method for SVAR models
under the zero-restrictions 
on the coefficient matrices,  
which correspond to sufficient conditions for 
LU decomposition
of the coefficient matrices of the reduced form 
of the SVAR models.
Moreover, we establish asymptotically normal estimators for the coefficient matrices and 
impulse responses, which enable us to 
construct test statistics for the 
simultaneous relationships of  
time-dependent data.
The finite-sample performance of the proposed method
is elucidated by numerical simulations.
We also present an example of an empirical study  
that analyzes the impact of policy rates on unemployment and prices.
\end{abstract}
\section{Introduction}\label{sec:introduction}
Structural vector autoregressive (SVAR) models
studied by \cite{Sims1980}, \cite{Bernanke1986}, and \cite{Blanchard1988} are widely used to 
analyze time-dependent macroeconomic data.
A basic SVAR model is expressed as follows:
\begin{equation}\label{eq:intro SVAR}
\bm{Y}_t = \bm{\mu} + \bm{A}_0 \bm{Y}_t + \sum_{s=1}^{p} \bm{A}_{s} \bm{Y}_{t-s} + \bm{v}_t,\quad
t \in \mathbb{Z},
\end{equation}
where $\bm{\mu} \in \mathbb{R}^k$ is an intercept, $\bm{A}_s, s=0,1,\ldots,p$ are $k \times k$
coefficient matrices, and $\{\bm{v}_t\}_{t \in \mathbb{Z}}$ is an innovation process.
The term $\bm{A}_0 \bm{Y}_t$ in the right-hand side represents the 
simultaneous relationships between 
the components of $\bm{Y}_t$.
When $\bm{I}_k - \bm{A}_0$ is non-singular, 
we have the following ordinal vector autoregressive 
(VAR) representation:
\begin{equation}\label{eq:intro reduced form}
\bm{Y}_t = \bm{\eta} + \sum_{s=1}^p \bm{B}_s \bm{Y}_{t-s} + 
\bm{e}_t,\quad t \in \mathbb{Z},
\end{equation}
where 
\[
\bm{\eta} = \bm{Q} \bm{\mu},\quad
 \bm{B}_s = \bm{Q} \bm{A}_s,\quad
s=1,\ldots,p,
\]
with $\bm{Q} = (\bm{I}_k - \bm{A}_0)^{-1}$, and 
$\bm{e}_t = \bm{Q} \bm{v}_t$.
We call \eqref{eq:intro reduced form} 
the reduced form of \eqref{eq:intro SVAR}.
Under some assumptions such as the stationarity 
of the process $\{\bm{Y}_t\}_{t \in \mathbb{Z}}$,
we can construct an asymptotically normal 
estimator for $\bm{B}_s, s=1,\ldots,p$ using methods 
such as 
ordinary least squares (OLS).
Meanwhile, some restrictions are required 
to recover 
$\bm{A}_s, s=0,\ldots,p$ 
from observable structures.
Such identification problems 
have been discussed by several researchers.
Typically, we consider zero restrictions, 
e.g., some specific components of $\bm{A}_s, s=0,1,\ldots,p$ are fixed to zero.
Under some additional assumptions, a general identification method based on zero 
restrictions 
was 
proposed by
\cite{Rubio-Ramirez2010}.
\cite{Canova2002} and \cite{Uhlig2005}
proposed identification methods for SVAR models
under the 
sign restriction of the coefficient matrices, i.e., 
the signs of some specific impulse 
responses 
are known.
\cite{Hyvarinen2010} and \cite{Lanne2017}
relaxed the restrictions on the 
coefficient matrices 
and developed other methods for non-Gaussian processes to allow flexible identification of  structures.

To correctly interpret the 
analysis of SVAR models,
it is necessary to impose appropriate restrictions
based on prior knowledge and data background.
Regarding this point, previous studies 
have focused on 
the restrictions 
on the coefficient matrices.
Meanwhile, it is often assumed that 
$\Var[\bm{v}_t]$ is a diagonal matrix. 
However, the components of $\bm{v}_t$ might be correlated if $\bm{v}_t$ includes some unobservable exogenous variables and does not correspond to the unique shocks of $\bm{Y}_t$.
Because an invalid assumption may 
lead to misunderstandings in causal interpretations, 
it is imperative to 
consider other identification methods under 
less restrictive conditions on the innovation process.

In this study, under mild conditions on the innovation process, where 
$\Var[\bm{v}_t]$ is allowed to be a 
non-diagonal matrix, we propose an identification method based on zero restrictions on the coefficient matrices 
and LU decomposition for 
a sub-matrix of $\bm{B} = (\bm{\eta}, \bm{B}_1,\ldots,\bm{B}_p)$.

Moreover, we establish 
asymptotically normal estimators for  
the coefficient matrices $\bm{A}_s, s=0,\ldots,p$
and impulse responses, enabling us 
to construct test statistics for 
the hypothesis testing
whose null hypothesis 
is that $\mathcal{H}_0: \bm{A}_0 = \bm{O}$, 
which can be used to 
verify whether the simultaneous relationships 
should be considered for the data $\{\bm{Y}_t\}_{t \in \mathbb{Z}}$.

The remainder of this article is organized as follows.
In Section \ref{sec:SVAR model setups}, 
we describe the model setup of SVAR models
and present a motivational example.
In Section \ref{sec:estimations}, 
we propose the identification and estimation 
methods for $\bm{A}_s, s=0,1,\ldots,p$
via LU decomposition 
under appropriate restrictions.
We also consider 
the impulse response estimation and  
hypothesis testing 
for $\bm{A}_0$ in this section.
The numerical simulations used to verify the 
asymptotic behavior of the estimators and test statistics are presented in  
Section \ref{sec:numerical}.
We further apply the proposed methods to 
analyze the impact of policy rates on employment 
and prices in this section.
The proofs of the main theoretical results are presented in Section \ref{sec:proofs}.
We discuss the causal interpretations
of the statistical inference for SVAR models
in the Appendix.
\section{SVAR models}\label{sec:SVAR model setups}
\subsection{Model setup and a motivational example}
We consider the following model:
\begin{equation}\label{eq:SVAR}
\bm{Y}_t = \bm{\mu} + \bm{A}_0 \bm{Y}_t + \sum_{s=1}^{p} \bm{A}_{s} \bm{Y}_{t-s} + \bm{v}_t,\quad
t \in \mathbb{Z},
\end{equation}
where $\bm{A}_s \in \mathbb{R}^{k \times k}, s=0,1,\ldots, p$ are the coefficient matrices and 
$\{\bm{v}_t\}_{t \in \mathbb{Z}}$ is an i.i.d. innovation 
process such that $\E[\bm{v}_t] = \bm{0}$.
We suppose that 
$\bm{A}_0$ is a $k \times k$ 
lower-triangular matrix
with all diagonal components being zero,
which means that 
$Y_{i, t}$ cannot be the direct cause
of $Y_{j, t}$ for $i >j$.
Several researchers have assumed that 
$\Var[\bm{v}_t]$ is a diagonal matrix.
Meanwhile, 
we consider the existence of 
contemporaneous confounding, 
indicating that 
$\bm{v}_{i, t}$ and $\bm{v}_{j, t}, i \neq j$
are correlated.
We now consider the following motivational example.
\begin{eg}\label{eg:SVAR-ex}
Let $\{{W}_t\}_{t \in \mathbb{Z}}$
be an $\mathbb{R}$-valued unobservable 
i.i.d. sequence.
We consider the following model:
\begin{equation}\label{eq:SVAR-ex}
\bm{Y}_t = \bm{A}_0 \bm{Y}_t + \sum_{s=1}^{p} \bm{A}_{s} \bm{Y}_{t-s} + \bm{A}_W{W}_t + \bm{u}_t,\quad
t \in \mathbb{Z},
\end{equation}
where $\bm{A}_s, s=1,\ldots,p \in \mathbb{R}^{3 \times 3}$, 
$\bm{A}_0 \in \mathbb{R}^{3 \times 3}$ is a 
lower-triangular matrix 
with diagonal elements being zero,
and 
$\{\bm{u}_t\}$ is an $\mathbb{R}^3$-valued 
i.i.d. sequence 
independent of ${W}_t$ with $\Var[\bm{u}_t] = \sigma_u^2 \bm{I}_3$, and 
$\bm{A}_W \in \mathbb{R}^{3 \times 1}$.
Thus, we can rewrite the model \eqref{eq:SVAR-ex} as follows:
\[
\bm{Y}_t = \bm{A}_0 \bm{Y}_t + \sum_{s=1}^{p} \bm{A}_{s} \bm{Y}_{t-s} + \bm{v}_t,\quad
t \in \mathbb{Z},
\]
where 
\[
\bm{v}_t =\bm{A}_W{W}_t + \bm{u}_t.
\]
Therefore, 
$\Var[\bm{v}_t]$ is not diagonal 
unless $\bm{A}_W \bm{A}_W^\top$ is diagonal.

The SVAR model can be 
regarded as a special case of 
linear structural equation models, as described in the Appendix.
For a linear structural equation model, 
we often consider a graphical representation.
Suppose that we observe 
$\bm{Y}_{1-p},\ldots,\bm{Y}_T$ for some $T \in \mathbb{N}$ and  
consider graph $\bm{G}$ with 
the vertex set  
$\tilde{\bm{V}} = \{Y_{i,s} : i =1,\ldots,k, s =1-p,\ldots,T\} \cup \{W_{s} :  s =1-p,\ldots,T \}$.
\begin{figure}[H]
    \centering
    \includegraphics[scale=0.5]{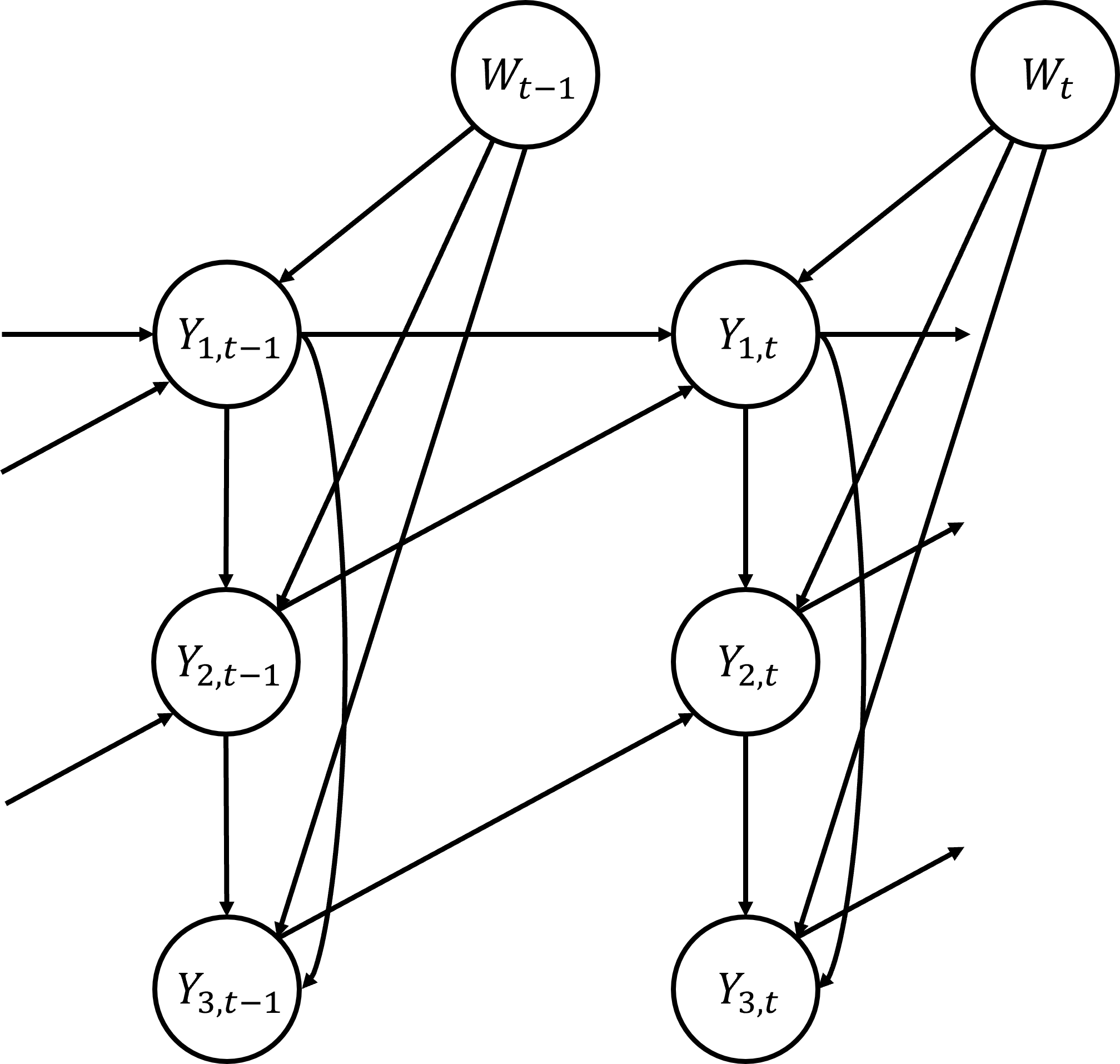}
    \caption{A graphical representation for the model \eqref{eq:SVAR-ex} with $p=1$}
    \label{fig:SVAR-ex}
\end{figure}
Figure \ref{fig:SVAR-ex} shows graph $\bm{G}$, 
where the edges represent the 
nonzero components 
of $\bm{A}_0, \bm{A}_1$ and 
$\bm{A}_W$.
Since the model structure does not depend on time, Figure \ref{fig:SVAR-ex} only shows
the components of $\bm{Y}_t$, their direct causes, $W_t$, and $W_{t-1}$
included in the vertex set.
Observe that $W_t, t=1-p,\ldots,T$ are (unobservable)
contemporaneous confounding factors, which 
implies correlations between the components of the 
innovation process $\{\bm{v}_t\}_{t \in \mathbb{Z}}$.
If we suppose that the 
covariance matrix of 
the innovation process is 
non-diagonal,
we cannot deal with such data 
using most of the conventional SVAR 
models.
\end{eg}
If the matrix $\bm{I}_k - \bm{A}_0$ is non-singular, then 
we have the reduced form of 
\eqref{eq:SVAR} as follows:
\begin{equation}\label{eq:SVAR-reduced}
\bm{Y}_t = \bm{Q} \bm{\mu} + \sum_{s=1}^{p} \bm{Q} \bm{A}_{s} \bm{Y}_{t-s} + \bm{Q}\bm{v}_t,
\end{equation}
where 
\[
\bm{Q} = (\bm{I}_k - \bm{A}_0)^{-1}. 
\]
This reduced form is a typical VAR model. 
Thus, to ensure 
the unique existence of 
a stationary and ergodic solution to 
\eqref{eq:SVAR}, 
it is sufficient to assume the following conditions.
\begin{assumption}\label{assump:stationary}
\begin{itemize}
\item[(i)]
$\bm{A}_0$ is a 
lower-triangle matrix with the diagonal elements being zero.
\item[(ii)]
The following polynomial
\[
\left| \bm{I}_{k} - \bm{Q}\bm{A}_1z - \bm{Q}\bm{A}_2z^2 - \cdots - \bm{Q}\bm{A}_pz^p \right|,\quad
z \in \mathbb{C}
\]
is nonzero for all $z \in \{\zeta \in \mathbb{C} : |\zeta| \leq 1\}$.
\end{itemize}
\end{assumption}
Condition (i) of Assumption \ref{assump:stationary}
guarantees that the model 
\eqref{eq:SVAR}
has a directed acyclic graph (DAG) structure.
Next, we consider the impulse response functions.
Consider the following VAR$(1)$ representation
of \eqref{eq:SVAR-reduced}:
\[
\bm{\xi}_t = \bm{\Lambda} \bm{\xi}_{t-1} + \bm{\epsilon}_t,
\]
where 
$\bm{\xi}_t = (\bm{Y}_t^\top,\ldots,\bm{Y}_{t-p+1}^\top)^\top$, 
$\bm{\epsilon}_t = (\bm{e}_t^\top,\bm{0}^\top,\ldots,\bm{0}^\top)^\top$, 
\[
\bm{\Lambda}
= \left(\begin{array}{ccccc}
\bm{B}_1 & \bm{B}_2 & \cdots & \bm{B}_{p-1} & \bm{B}_p \\
\bm{I}_k & \bm{O} & \cdots & \bm{O} & \bm{O} \\
\bm{O} & \bm{I}_k & \cdots & \bm{O} & \bm{O} \\
\vdots & \vdots & \ddots & \vdots & \vdots \\
\bm{O} & \bm{O} & \cdots & \bm{I}_k & \bm{O} \\
\end{array}
\right),
\]
and $\bm{B}_s = \bm{Q}\bm{A}_s$ for $s=1,\ldots,p$. 
We omit $\bm{\mu}$ here for simplicity because impulse responses do not depend on it.
Under Assumption \ref{assump:stationary}, 
we have the following moving average (MA) representation:
\begin{equation}\label{eq:SVAR-MA}
\bm{\xi}_t
= \sum_{s=0}^\infty \bm{\Lambda}^s \bm{\epsilon}_{t-s}.
\end{equation}
For the top $k$ rows in 
\eqref{eq:SVAR-MA}, 
there exist 
$\bm{\Psi}_0 = \bm{I}_k$ and 
$\bm{\Psi}_s \in \mathbb{R}^{k \times k}, s=1,2,\ldots$ such that 
\[
\bm{Y}_t
= \sum_{s=0}^\infty \bm{\Psi}_s \bm{e}_{t-s}
= \sum_{s=0}^\infty \bm{\Psi}_s \tilde{\bm{L}} \tilde{\bm{u}}_{t-s},
\]
where 
$\tilde{\bm{L}} \in \mathbb{R}^{k \times k}$
is a
lower-unitriangular 
matrix obtained by 
the LU decomposition of $\Var[\bm{e}_t]$
and 
$\tilde{\bm{u}}_t = \tilde{\bm{L}}^{-1}\bm{e}_t$.
Notably, $\bm{\Psi}_s$ coincides with the 
$k \times k$ sub-matrix of $\bm{\Lambda}^s$
corresponding to the first $k$ columns and rows.
Therefore, the orthogonalized impulse response
$\mathrm{OIRF}_{ij}(s)$ and 
non-orthogonalized impulse response $\mathrm{IRF}_{ij}(s)$ are, respectively, given by
\[
\mathrm{OIRF}_{ij}(s) = \left(
\bm{\Psi}_s \tilde{\bm{L}}
\right)_{ij},\quad
\mathrm{IRF}_{ij}(s) = \left(
\bm{\Psi}_s
\right)_{ij}.
\]
By considering the SVAR model as a 
special case of linear structural equation models, 
the coefficient matrices and 
impulse responses 
can be regarded as causal effects
(see the Appendix for the detail of such interpretations).
\section{Identification and statistical inference of  SVAR models}\label{sec:estimations}
\subsection{Estimations of the coefficient matrices}\label{subsec:coef}
In this section, we establish asymptotically  normal estimators for 
the coefficient matrices of 
model \eqref{eq:SVAR}.

We begin by considering 
the estimation method for the 
reduced form \eqref{eq:SVAR-reduced} of 
\eqref{eq:SVAR}.
We rewrite the model as follows:
\begin{equation}\label{eq:SVAR-structure sec4}
\bm{Y}_t = \bm{A}_0 \bm{Y}_t + \bm{A} \bm{X}_{t-1}
+ \bm{v}_t, 
\end{equation}
where 
$\bm{A} = (\bm{\mu}, \bm{A}_1,\ldots, \bm{A}_p)$
and 
$\bm{X}_{t-1} = (1, \bm{Y}_{t-1}^\top,\ldots,\bm{Y}_{t-p}^\top)^\top$.
The reduced form 
can be represented as follows.
\begin{equation}\label{eq:SVAR-reduced sec4}
\bm{Y}_t = \bm{B} \bm{X}_{t-1} + \bm{e}_t,
\end{equation}
where 
$\bm{B} = \bm{Q} \bm{A}$ and 
$\bm{e}_t = \bm{Q}\bm{v}_t$.
Suppose that we observe $(\bm{Y}_{1-p},\ldots,\bm{Y}_T)$.
We define the data matrix 
$\bm{Y}$ and the design matrix 
$\bm{X}$ as follows:
\[
\bm{Y} = (\bm{Y}_1,\ldots,\bm{Y}_T)^\top\quad
\mbox{and}\quad
\bm{X} = (\bm{X}_0,\ldots,\bm{X}_{T-1})^\top,
\]
respectively.
Let $r=1+kp$ and 
$\bm{\Theta} \subset \mathbb{R}^{k \times r}$
be a compact parameter space of $\bm{B}$
and $\bm{B}_*$ be the true value of $\bm{B}$.
Then, we consider the following least squares estimators $\hat{\bm{B}}_T$ and $\hat{\bm{b}}_T$
for $\bm{B}$ and $\ve(\bm{B})$:
\[
\hat{\bm{B}}_T = \bm{Y}^\top \bm{X} (\bm{X}^\top \bm{X})^{-1}
= \left(
\sum_{t=1}^T \bm{Y}_t \bm{X}_{t-1}^\top
\right)
\left(
\sum_{t=1}^T \bm{X}_{t-1} \bm{X}_{t-1}^\top
\right)^{-1},
\]
and 
\[
\hat{\bm{b}}_T
= \ve(\hat{\bm{B}}_T)
= \bm{b}_* + \left\{\left(
\frac{1}{T}
\sum_{t=1}^T \bm{X}_{t-1} \bm{X}_{t-1}^\top
\right)^{-1} \otimes \bm{I}_k\right\}
\left(
\frac{1}{T} \sum_{t=1}^T \bm{X}_{t-1} \otimes \bm{e}_t
\right),
\]
where $\bm{b}_* = \ve(\bm{B}_*)$.
To establish the asymptotic behavior of the 
estimator $\hat{\bm{b}}_T$, we assume the following conditions.
\begin{assumption}\label{assump:regularity LSE}
\begin{itemize}
\item[(i)]
It holds that
\[
\E[\|\bm{v}_t\|_2^4] < \infty,\quad
t \in \mathbb{Z}.
\]
\item[(ii)]
$\Var[\bm{v}_t]$ and 
$\E[\bm{X}_{t-1} \bm{X}_{t-1}^\top]$ are positive definite.
\item[(iii)]
The true value $\bm{B}_*$ is an interior point in  
the parameter space $\bm{\Theta} \subset \mathbb{R}^{k \times r}$.
\end{itemize}
\end{assumption}
Then, we have the following asymptotic normality 
of $\hat{\bm{b}}_T$.
\begin{prop}\label{prop:normality LSE B}
Under Assumptions \ref{assump:stationary} and \ref{assump:regularity LSE}, it holds that 
\[
\sqrt{T}(\hat{\bm{b}}_T - \bm{b}_*)
\to^d N(\bm{0}, \bm{\Gamma}^{-1} \otimes \bm{\Sigma}),\quad
T \to \infty,
\]
where 
$\bm{\Sigma} = \E[\bm{e}_t \bm{e}_t^\top]$
and $\bm{\Gamma} = \E[\bm{X}_{t-1} \bm{X}_{t-1}^\top]$.
\end{prop}
The proof can be found in, e.g., 
\cite{Lutkepohl2013} or \cite{Hamilton1994},
therefore, we omit it.

Next, we introduce 
estimators for $\bm{Q}$, $\bm{A}_0$, and $\bm{A}$ 
using LU decomposition.
Denote 
\[
\bm{A} = (\bm{a}_1,\bm{a}_2,\ldots,\bm{a}_r).
\]
where $\bm{a}_1= \bm{\mu}$.
The following condition is  
sufficient to construct 
estimators for $\bm{Q}$, $\bm{A}_0$, and $\bm{A}$ 
via LU decomposition
and derive its asymptotic behavior.
\begin{assumption}\label{assump:regularity A}
There exists a (known) $k$-tuple of distinct 
positive numbers $(j_1,\ldots,j_k)$ such that 
$(\bm{a}_{j_1 *},\ldots,\bm{a}_{j_k *})$
is a non-singular 
upper-triangular matrix, 
where 
$\bm{A}_*$ is the true value of $\bm{A}$.
\end{assumption}
Let $g: \mathbb{R}^{k \times r} \to \mathbb{R}^{k \times k}$ be the map 
defined as follows 
\[
g(\bm{C}) = (\bm{c}_{j_1},\ldots,\bm{c}_{j_k}),
\]
where $\bm{c}_m$ is the $m$-th column of the 
matrix $\bm{C}$.
Similarly to
\cite{Rubio-Ramirez2010},  
we should impose zero restrictions for 
lower-triangular
part of $g(\bm{A})$ based on 
the data background
(see Section \ref{subsec:real data} 
for a concrete example).

Note that 
$g(\bm{B}) = \bm{Q} g(\bm{A})$. 
For a nonsingular matrix 
$\bm{C} \in \mathbb{R}^{k \times k}$ 
which allows 
an LU-decomposition with 
a lower-unitriangular matrix, 
we introduce the following notation:
\[
\bm{C} = \bm{L}(\bm{C}) \bm{U}(\bm{C}).
\]
Let $\bm{q}, \bm{a}_0$, and $\bm{a}$ be the  vectorizations  
of $\bm{Q}, \bm{A}_0$, and $\bm{A}$, respectively, 
i.e., 
\[
\bm{q} = \ve(\bm{Q}),\quad
\bm{a}_0 = \ve(\bm{A}_0),\quad
\mbox{and}
\quad
\bm{a} = \ve(\bm{A}).
\]
We define the estimators for $\bm{q}, \bm{a}_0$, and $\bm{a}$ as follows.
\[
\hat{\bm{q}}_T 
= f_1(\hat{\bm{b}}_T) := \ve(\bm{L}_g(\hat{\bm{B}}_T)),
\]
\[
\hat{\bm{a}}_{0T} 
= f_2(\hat{\bm{b}}_T) := \ve(\bm{I}_k - \bm{L}_g(\hat{\bm{B}}_T)^{-1}),
\]
and 
\[
\hat{\bm{a}}_{T} 
= f_3(\hat{\bm{b}}_T) := \left(\bm{I}_r \otimes 
\bm{L}_g(\hat{\bm{B}}_T)\right)^{-1}\hat{\bm{b}}_T,
\]
where $\bm{L}_g := \bm{L} \circ g$.
The asymptotic normality of 
the estimators follows from the delta method.
\begin{thm}\label{thm:normality A}
Suppose that Assumptions \ref{assump:stationary}, 
\ref{assump:regularity LSE}, and \ref{assump:regularity A} hold.
Then, the following conditions hold:
\begin{equation}\label{eq:normality q}
\sqrt{T}(\hat{\bm{q}}_T - \bm{q}_*)
\to^d N(\bm{0}, \bm{\Sigma}_1),
\end{equation}
\begin{equation}\label{eq:normality ao}
\sqrt{T}(\hat{\bm{a}}_{0T} - \bm{a}_{0*})
\to^d N(\bm{0}, \bm{\Sigma}_2),
\end{equation}
\begin{equation}\label{eq:normality a}
\sqrt{T}(\hat{\bm{a}}_T - \bm{a}_*)
\to^d N(\bm{0}, \bm{\Sigma}_3),
\end{equation}
as $T \to \infty$,
where 
\[
\bm{\Sigma}_l = \bm{J}_l \bm{\Sigma}_b \bm{J}_l^\top
\]
with 
$\bm{\Sigma}_b = \bm{\Gamma}^{-1} \otimes \bm{\Sigma}$ and  
\[
\bm{J}_l = \left.\frac{\partial f_l(\bm{b})}{\partial \bm{b}}\right|_{\bm{b} = \bm{b}_*},\quad
l=1,2,3.
\]
\end{thm}
The asymptotic covariance matrices 
are singular because some components are fixed at $0$ or $1$ by LU decomposition.
\subsection{Estimation of impulse responses}\label{subsec:impulse}
Next, we construct estimators for the impulse response $\mathrm{IRF}_{ij}(h)$ for 
$i, j = 1,2,\ldots,k$ and $h>0$.
The matrix $\bm{\Psi}_h = (\mathrm{IRF}_{ij}(h))_{i, j = 1,\ldots,k}$ satisfies 
\[
\bm{\Psi}_h = [\bm{\Lambda}^h]_k^k,
\] 
where $[\bm{\Lambda}^h]_k^k$ is the $k \times k$ sub-matrix of $\bm{\Lambda}^h$ corresponding to the top $k$ columns and rows.
Therefore, there exists a differentiable map 
$f_{4, h}$ such that 
$f_{4, h}(\bm{b}) = \bm{\psi}_h$, where 
$\bm{\psi}_h = \ve(\bm{\Psi}_h)$.
We define an estimator for $\bm{\psi}_h$
by $\hat{\bm{\psi}}_{h T} = f_{4, h}(\hat{\bm{b}}_T)$.
The following proposition was proved by 
\cite{lutkepohl1990}.
\begin{prop}\label{prop:estimator psi}
Suppose that assumptions \ref{assump:stationary}, 
\ref{assump:regularity LSE}, and \ref{assump:regularity A} hold.
For every $h >0$, it holds that 
\begin{equation}\label{eq:normality psi}
\sqrt{T}(\hat{\bm{\psi}}_{hT} - \bm{\psi}_{h*})
\to^d N(\bm{0}, \bm{\Sigma}_{4, h}),\quad
T \to \infty,
\end{equation}
where 
$\bm{\psi}_{h*} = \ve([\bm{\Lambda}_*^h]_k^k)$,
\[
\bm{\Lambda}_*
=\left(\begin{array}{ccccc}
\bm{B}_{1*} & \bm{B}_{2*} & \cdots & \bm{B}_{p-1*} & \bm{B}_{p*} \\
\bm{I}_k & \bm{O} & \cdots & \bm{O} & \bm{O} \\
\bm{O} & \bm{I}_k & \cdots & \bm{O} & \bm{O} \\
\vdots & \vdots & \ddots & \vdots & \vdots \\
\bm{O} & \bm{O} & \cdots & \bm{I}_k & \bm{O} \\
\end{array}
\right),
\]
and $\bm{\Sigma}_{4, h} = \bm{J}_{4, h}\bm{\Sigma}_{\bm{b}} \bm{J}_{4, h}^\top$ with 
\[
\bm{J}_{4, h} = \left.\frac{\partial f_{4, h}(\bm{b})}{\partial \bm{b}}\right|_{\bm{b} = \bm{b}_*}.
\]
\end{prop}
The general expression of 
each component of $\bm{J}_{4, h}$ can be found in 
\cite{lutkepohl1990}. 

Let $\bm{\Psi}_{h*}^{\mathrm{o}}, h>0$ 
be a matrix defined by 
\[
\bm{\Psi}_{h*}^{\mathrm{o}}
= \bm{\Psi}_{h*}\bm{Q}_{*},
\]
which corresponds to a ``total effect'' for a 
linear structural equation model as described in 
the Appendix.
We define the estimator for 
$\bm{\psi}_{h*, T}^{\mathrm{o}}= \ve(\bm{\Psi}_{h*}\bm{Q}_{*})$ 
as follows.
\[
\hat{\bm{\psi}}_{hT}^{\mathrm{o}}
= f_{5, h}(\hat{\bm{b}}_T)
:= \left(
\bm{L}_g(\hat{\bm{B}}_T)^\top \otimes \bm{I}_k
\right) f_{4, h}(\hat{\bm{b}}_T).
\]
Because the map $f_{5, h}$ for every $h$ 
is differentiable, 
we obtain the following proposition based on the delta method.
\begin{prop}\label{prop:estimator total effect}
Suppose that Assumptions \ref{assump:stationary}, 
\ref{assump:regularity LSE}, and \ref{assump:regularity A} hold.
Then, it holds that 
\begin{equation}\label{eq:normality total effect}
\sqrt{T}(\hat{\bm{\psi}}_{hT}^{\mathrm{o}} - \bm{\psi}_{h *}^{\mathrm{o}})
\to^d N(\bm{0}, \bm{\Sigma}_{5, h})
,\quad T \to \infty,
\end{equation}
where 
$\bm{\psi}_{h *}^{\mathrm{o}} = \ve(\bm{\Psi}_{h *} \bm{Q}_*)$, 
$\bm{Q}_*$ is the true value of $\bm{Q}$, and
$\bm{\Sigma}_{5, h} = \bm{J}_{5, h}\bm{\Sigma}_{\bm{b}} \bm{J}_{5, h}^\top$ with 
\[
\bm{J}_{5, h} = \left.\frac{\partial f_5(\bm{b})}{\partial \bm{b}}\right|_{\bm{b} = \bm{b}_*}.
\]
\end{prop}
\subsection{Tests for simultaneous relationships}\label{subsec:test}
In this section, we consider the following 
test:
\begin{equation}\label{eq:test A0}
\mathcal{H}_0:\ 
\bm{A}_0 = O,\quad
\mathcal{H}_1:\ \bm{A}_0 \neq O.
\end{equation}
Let $\hat{\bm{e}}_t = \bm{Y}_t - \hat{\bm{B}}_T \bm{X}_{t-1}$ and 
\[
\hat{\bm{\Sigma}}_{bT}
= \left(
\frac{1}{T} \sum_{t=1}^T \bm{X}_{t-1} \bm{X}_{t-1}^\top
\right)^{-1}
\otimes \left(
\frac{1}{T} \sum_{t=1}^T \hat{\bm{e}}_t \hat{\bm{e}}_t^\top
\right).
\]
The estimators 
for $\bm{\Sigma}_l, l=1,2,3$ and $\bm{\Sigma}_{l, h}, l=4, 5$ can be constructed as follows:
\[
\hat{\bm{\Sigma}}_l
= \hat{\bm{J}}_l \hat{\bm{\Sigma}}_{\bm{b}}\hat{\bm{J}}_l^\top,\quad l=1,2,3,
\]
and 
\[
\hat{\bm{\Sigma}}_{l, h}
= \hat{\bm{J}}_{l, h} \hat{\bm{\Sigma}}_{\bm{b}}\hat{\bm{J}}_{l, h}^\top,\quad l=4, 5
\]
with 
\[
\hat{\bm{J}}_{l T}=\left.\frac{\partial f_l(\bm{b})}{\partial \bm{b}}\right|_{\bm{b} = \hat{\bm{b}}_{lT}},\quad
l=1,2,3
\]
and 
\[
\hat{\bm{J}}_{l, h T}=\left.\frac{\partial f_{l, h}(\bm{b})}{\partial \bm{b}}\right|_{\bm{b} = \hat{\bm{b}}_{lT}},\quad
l=4, 5.
\]
The following lemma is obtained from 
the asymptotic normality of the estimators.
\begin{lem}\label{lem:estimation sub-vector}
Suppose that Assumptions \ref{assump:stationary}, 
\ref{assump:regularity LSE}, and \ref{assump:regularity A} hold.
For every 
$\bm{w}_1, \bm{w}_2 \in \mathbb{R}^{k^2}$, 
and $\bm{w}_3 \in \mathbb{R}^{k(1+kp)}$,
let 
\[
z_{1T}(\bm{w}_{1})
:= \sqrt{T} (\bm{w}_{1}^\top \hat{\bm{\Sigma}}_{1} \bm{w}_{1})^{-1/2} \bm{w}_{1}^\top
(\hat{\bm{q}}_T - \bm{q}_*),
\]
\[
z_{2T}(\bm{w}_{2})
:= \sqrt{T} (\bm{w}_{2}^\top \hat{\bm{\Sigma}}_{2} \bm{w}_{2})^{-1/2} \bm{w}_{2}^\top
(\hat{\bm{a}}_{0T} - \bm{a}_{0*}).
\]
and 
\[
z_{3T}(\bm{w}_{3}) := \sqrt{T} (\bm{w}_{3}^\top \hat{\bm{\Sigma}}_{\bm{b}} \bm{w}_{3})^{-1/2} \bm{w}_{3}^\top
(\hat{\bm{b}}_T - \bm{b}_*).
\]
Then, it holds that 
\begin{equation}\label{eq:normality comp b}
z_{lT}(\bm{w}_{l})
\to^d N(0, 1), \quad T \to \infty,\quad
l=1,2,3.
\end{equation}
\end{lem}
Under the null hypothesis 
$\mathcal{H}_0$, 
we have 
$
\bm{Q}_* = \bm{I}_k$, $\bm{A}_{0*} = \bm{O}$, 
and $g(\bm{B}_*) = g(\bm{A}_*)$ is a 
non-singular 
upper-triangular
matrix.
For a matrix 
$\bm{C} \in \mathbb{R}^{k\times (1+kp)}$, 
$g(\bm{C}) \in \mathbb{R}^{k \times k}$ is a sub-matrix of $\bm{C}$.
We define the sub-vectors 
$\bm{q}_{\mathrm{sub}}$ and $\bm{a}_{0 \mathrm{sub}}$
of $\bm{q}$ and $\bm{a}_0$, 
which comprise the lower-triangular components 
of 
$\bm{Q}$ and $\bm{A}_0$, except for diagonal components, respectively.
Similarly, let $\bm{\beta}_{\mathrm{sub}}$
be a sub-vector which comprises
the lower-triangular
components of 
$g(\bm{B})$, except for diagonal components. 
Then, we have 
\[
\bm{w}^\top \bm{q}_{*\mathrm{sub}}
= \bm{w}^\top \bm{a}_{0*\mathrm{sub}}
= \bm{w}^\top \bm{\beta}_{*\mathrm{sub}} =0
\]
for every vector $\bm{w} \in \mathbb{R}^{k(k-1)/2}$.
Thus, we obtain the following theorem.
\begin{thm}\label{thm:test}
Suppose that Assumptions \ref{assump:stationary}, 
\ref{assump:regularity LSE}, and \ref{assump:regularity A} hold.
For $l=1,2,3$ and 
$\bm{v} \in \mathbb{R}^{k(k-1)/2}$, 
consider the following test statistics
$z_{lT}$
\[
z_{1T} = \sqrt{T} (\bm{v}^\top \hat{\bm{\Sigma}}_{1\mathrm{sub}} \bm{v})^{-1/2} \bm{v}^\top \hat{\bm{q}}_{T\mathrm{sub}},
\]
\[
z_{2T} = \sqrt{T} (\bm{v}^\top \hat{\bm{\Sigma}}_{2\mathrm{sub}} \bm{v})^{-1/2} \bm{v}^\top \hat{\bm{a}}_{0 T\mathrm{sub}},
\]
and 
\[
z_{3T} = \sqrt{T} (\bm{v}^\top \hat{\bm{\Sigma}}_{\bm{b}\mathrm{sub}} \bm{v})^{-1/2} \bm{v}^\top \hat{\bm{\beta}}_{T\mathrm{sub}},
\]
where 
$\hat{\bm{\Sigma}}_{1\mathrm{sub}}, \hat{\bm{\Sigma}}_{2\mathrm{sub}},$ and 
$\hat{\bm{\Sigma}}_{\bm{b}\mathrm{sub}}$ are 
sub-matrices of 
$\hat{\bm{\Sigma}}_{1}, \hat{\bm{\Sigma}}_{2},$ and 
$\hat{\bm{\Sigma}}_{\bm{b}}$
corresponding to 
sub-vectors $\hat{\bm{q}}_{T\mathrm{sub}}$, 
$\hat{\bm{a}}_{0T\mathrm{sub}}$, and 
$\hat{\bm{\beta}}_{T\mathrm{sub}}$, respectively.
\begin{itemize}
\item[(i)]
Under the null hypothesis $\mathcal{H}_0$, 
\[
z_{lT} \to^d N(0, 1),\quad T \to \infty,\quad
l=1,2,3.
\]
\item[(ii)]
Assume in addition that 
$
\bm{v}_1^\top \bm{q}_{*}, 
\bm{v}_2^\top \bm{a}_{0*}, 
\bm{v}_3^\top \bm{b}_*
\neq 0
$
under the alternative hypothesis $\mathcal{H}_1$, 
then, 
\[
\Prob(|z_{lT}| > c) \to 1,\quad T \to \infty,\quad
l=1,2,3
\]
for every $c>0$.
\end{itemize}
\end{thm}
\begin{rem}\label{rem:asymptotic theory}
Similarly to the example of the empirical study described in the next section,
the number of observations 
of time series data that seems to satisfy stationarity tends to be small. However, SVAR models have 
at least $k^2p$ 
parameters, 
and we often analyze the data such that the $k^2p/T \not\asymp 0$. 
To handle such data, we should consider 
central limit theorems or Gaussian approximations under 
high-dimensional settings, e.g.,  
the case where $k^2p/T \to \kappa \in (0, 1)$ as 
$T \to \infty$.
See \cite{Koike2021},  
\cite{Chernozhukov2023} for the
recent development of high-dimensional 
asymptotic theory,
\cite{Fang2023} for the degenerate case, and 
\cite{Belloni2018} for martingale high-dimensional 
central limit theorems.
\end{rem}
\section{Numerical studies}\label{sec:numerical}
\subsection{Simulations}
In this section, we verify the finite-sample performances of the proposed estimators and test statistics.
Because each component of 
$f_l(\bm{b})$ and $f_{l, h}(\bm{b})$ for $l=1,\ldots,5$
can be represented as a
rational function of the components of $\bm{b}$, 
we can calculate their derivatives using, e.g., 
``PyTorch'' in Python.
In the sequel, we consider a stationary model 
with $k=p=5$.
The innovation process $\{\bm{v}_t\}_{t \in \mathbb{Z}}$ is generated by the following model:
\[
\bm{v}_t = \bm{A}_W \bm{W}_t + \bm{u}_t,
\] 
where $\{\bm{W}_t\}_{t \in \mathbb{Z}}$, 
$\{\bm{u}_t\}_{t \in \mathbb{Z}}$ are two and five-dimensional mutually uncorrelated 
i.i.d. Laplace distributed with mean $\bm{0}$
and variance $0.5$, respectively, and 
\[
\bm{A}_W = 
\begin{pmatrix}
    0.5 & -0.5 \\
    0.5 &  0.5 \\
   -0.5 &  0.5 \\
    0.4 &  0.6 \\
   -0.4 & -0.6 \\
\end{pmatrix}.
\]
The true values of the coefficient matrices are 
shown with the heat maps (Figure \ref{fig:parameter simulation}), 
where the $x$-axis of the figure for $g(\bm{A})$
represents the column number of $\bm{A}$.
\begin{figure}[H]
    \centering
    \includegraphics[scale=0.6]{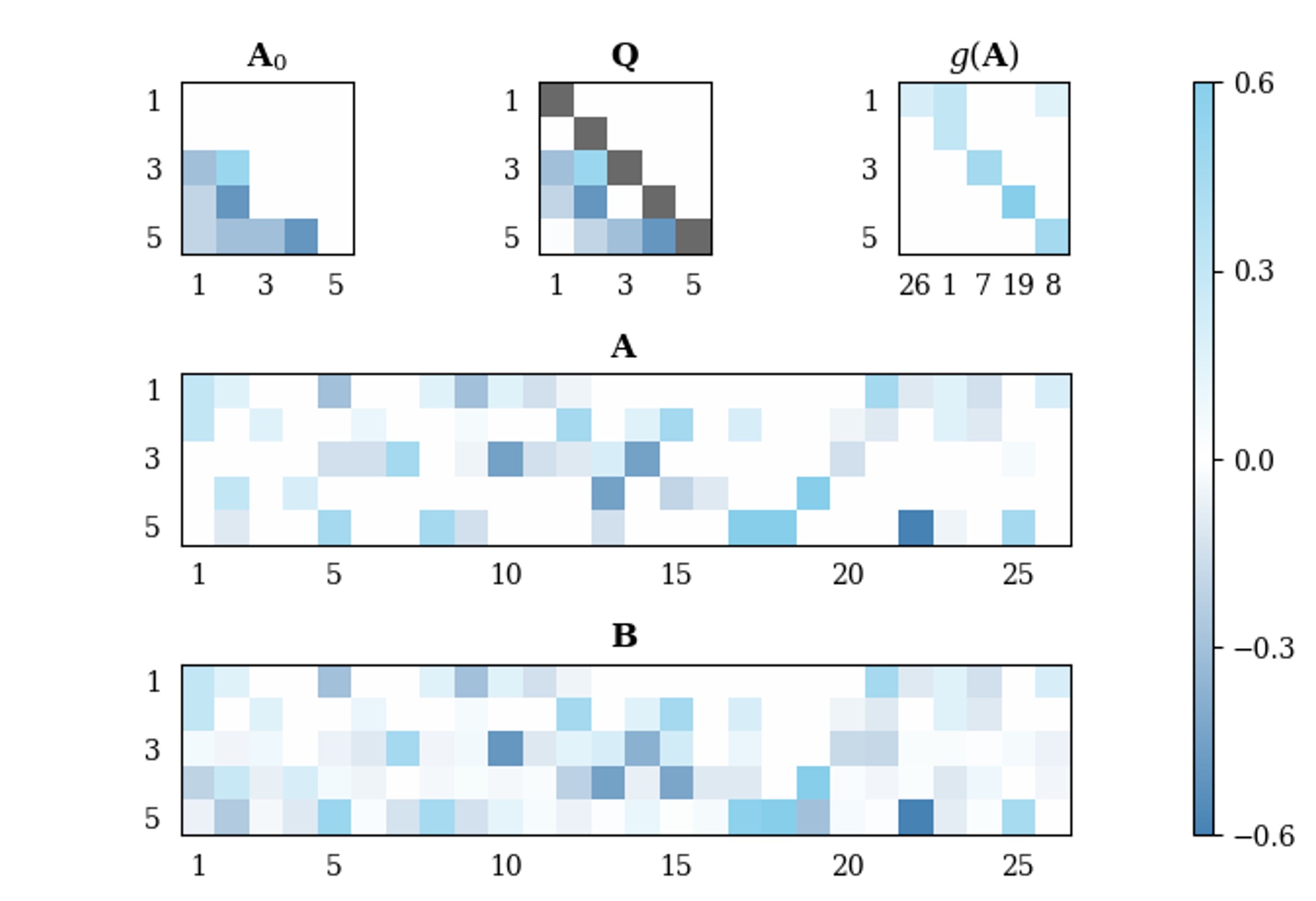}
    \caption{Heat maps of the coefficient matrices}
    \label{fig:parameter simulation}
\end{figure}
We calculate the proposed 
estimators and test statistics 
for $1000$ replications.

First, we evaluate the performances of the 
estimators $\hat{\bm{b}}_T$, $\hat{\bm{q}}_T$, 
$\hat{\bm{a}}_{0T}$, and $\hat{\bm{a}}_T$, 
using the mean of empirical bias (MB)
and mean of empirical mean absolute error (MMAE).
Tables \ref{table:est1}--\ref{table:est3}
summarize the results.
\begin{table}[H]
  \caption{Estimators of the coefficient matrices $(T=100)$}
  \label{table:est1}
  \centering
  \begin{tabular}{lcccc}
    \hline
      & $\hat{\bm{b}}_T$  &  $\hat{\bm{q}}_T$ & $\hat{\bm{a}}_{0T}$ & $\hat{\bm{a}}_T$  \\
    \hline \hline
    MB  & 0.008  & 0.062 & 0.065 & 0.026  \\
    MMAE   & 0.11 & 0.155 & 0.16 & 0.138 \\
    \hline
  \end{tabular}
\end{table}

\begin{table}[H]
  \caption{Estimators of the coefficient matrices $(T=200)$}
  \label{table:est2}
  \centering
  \begin{tabular}{lcccc}
    \hline
      & $\hat{\bm{b}}_T$  &  $\hat{\bm{q}}_T$ & $\hat{\bm{a}}_{0T}$ & $\hat{\bm{a}}_T$  \\
    \hline \hline
    MB  & 0.004  & 0.043 & 0.049 & 0.019  \\
    MMAE   & 0.072 & 0.124 & 0.127 & 0.092 \\
    \hline
  \end{tabular}
\end{table}

\begin{table}[H]
  \caption{Estimators of the coefficient matrices $(T=500)$}
  \label{table:est3}
  \centering
  \begin{tabular}{lcccc}
    \hline
      & $\hat{\bm{b}}_T$  &  $\hat{\bm{q}}_T$ & $\hat{\bm{a}}_{0T}$ & $\hat{\bm{a}}_T$  \\
    \hline \hline
    MB  & 0.002  & 0.019 & 0.019 & 0.008  \\
    MMAE   & 0.043 & 0.079 & 0.077 & 0.053 \\
    \hline
  \end{tabular}
\end{table}
We also calculate the estimators for impulse responses for $h=1, 2, 3$.
The results are summarized in 
Table \ref{table:est4}--\ref{table:est6}.
In summary, we can observe the consistency of the proposed estimators.
\begin{table}[H]
  \caption{Estimators for (controlled) total effects $(T=100)$}
  \label{table:est4}
  \centering
  \begin{tabular}{lcccccc}
    \hline
      & $\hat{\bm{\psi}}_{1T}$  &  $\hat{\bm{\psi}}_{2T}$ & $\hat{\bm{\psi}}_{3T}$ & $\hat{\bm{\psi}}_{1T}^{\mathrm{o}}$ & $\hat{\bm{\psi}}_{2T}^{\mathrm{o}}$ & $\hat{\bm{\psi}}_{3T}^{\mathrm{o}}$   \\
    \hline \hline
    MB  & 0.008  & 0.009 & 0.006 & 0.030 & 0.031 & 0.032 \\
    MMAE   & 0.118 & 0.105 & 0.102 & 0.145 & 0.138 & 0.143 \\
    \hline
  \end{tabular}
\end{table}

\begin{table}[H]
  \caption{Estimators for (controlled) total effects $(T=200)$}
  \label{table:est5}
  \centering
  \begin{tabular}{lcccccc}
    \hline
      & $\hat{\bm{\psi}}_{1T}$  &  $\hat{\bm{\psi}}_{2T}$ & $\hat{\bm{\psi}}_{3T}$ & $\hat{\bm{\psi}}_{1T}^{\mathrm{o}}$ & $\hat{\bm{\psi}}_{2T}^{\mathrm{o}}$ & $\hat{\bm{\psi}}_{3T}^{\mathrm{o}}$   \\
    \hline \hline
    MB  & 0.004  & 0.003 & 0.004 & 0.019 & 0.023 & 0.022 \\
    MMAE   & 0.076 & 0.069 & 0.066 & 0.095 & 0.094 & 0.099 \\
    \hline
  \end{tabular}
\end{table}

\begin{table}[H]
  \caption{Estimators for (controlled) total effects $(T=500)$}
  \label{table:est6}
  \centering
  \begin{tabular}{lcccccc}
    \hline
      & $\hat{\bm{\psi}}_{1T}$  &  $\hat{\bm{\psi}}_{2T}$ & $\hat{\bm{\psi}}_{3T}$ & $\hat{\bm{\psi}}_{1T}^{\mathrm{o}}$ & $\hat{\bm{\psi}}_{2T}^{\mathrm{o}}$ & $\hat{\bm{\psi}}_{3T}^{\mathrm{o}}$   \\
    \hline \hline
    MB  & 0.002  & 0.002 & 0.002 & 0.007 & 0.008 & 0.008 \\
    MMAE   & 0.047 & 0.042 & 0.040 & 0.057 & 0.055 & 0.062 \\
    \hline
  \end{tabular}
\end{table}
Next, we show the asymptotic normality
of the proposed estimators
$\hat{\bm{q}}_T, \hat{\bm{a}}_{0T}, \hat{\bm{a}}_T$, 
and $\hat{\bm{\psi}}_{h, T}$ for 
$h=1,2,3$.
To do this, we introduce the following 
random variables.
\begin{align*}
    s_{1T} &= \sqrt{T}(\bm{1}^{\top}_{k^2}\hat{\bm{\Sigma}}_{1}\bm{1}_{k^2})^{-\frac{1}{2}}\bm{1}^{\top}_{k^2}(\hat{\bm{q}}_T-\bm{q}_{*}) \\
    s_{2T} &= \sqrt{T}(\bm{1}^{\top}_{k^2}\hat{\bm{\Sigma}}_{2}\bm{1}_{k^2})^{-\frac{1}{2}}\bm{1}^{\top}_{k^2}(\hat{\bm{a}}_{0T}-\bm{a}_{0*}) \\
    s_{3T} &= \sqrt{T}(\bm{1}^{\top}_{kr}\hat{\bm{\Sigma}}_{3}\bm{1}_{kr})^{-\frac{1}{2}}\bm{1}^{\top}_{kr}(\hat{\bm{a}}_{T}-\bm{a}_{*}) \\
    s_{5,hT} &= \sqrt{T}(\bm{1}^{\top}_{k^2}\hat{\bm{\Sigma}}_{5,h}\bm{1}_{k^2})^{-\frac{1}{2}}\bm{1}^{\top}_{k^2}(\hat{\bm{\psi}}_{hT}^{\mathrm{o}}-\bm{\psi}_{h*}^{\mathrm{o}}) ,
\end{align*}
where $\bm{1}_{k^2}=(1,1,\ldots,1)^\top \in \mathbb{R}^{k^2}$ and $\bm{1}_{kr}=(1,1,\ldots,1)^\top \in \mathbb{R}^{kr}$.
Figures \ref{fig:hist1}--\ref{fig:hist6}
show the histograms of the random variables over 
1000 replications;
the $y$-axis indicates the relative frequency of the 
statistics over the replications
and the dashed line is the density function of the 
standard normal distribution.
Moreover, Tables \ref{table:size1} and 
\ref{table:size2} show the 
empirical tail probabilities of the 
random variables.
Even for the relatively small observations $T=100$, 
the estimators are well approximated by 
the normal distribution. 
\begin{figure}[H]
    \centering
    \includegraphics[scale=0.5]{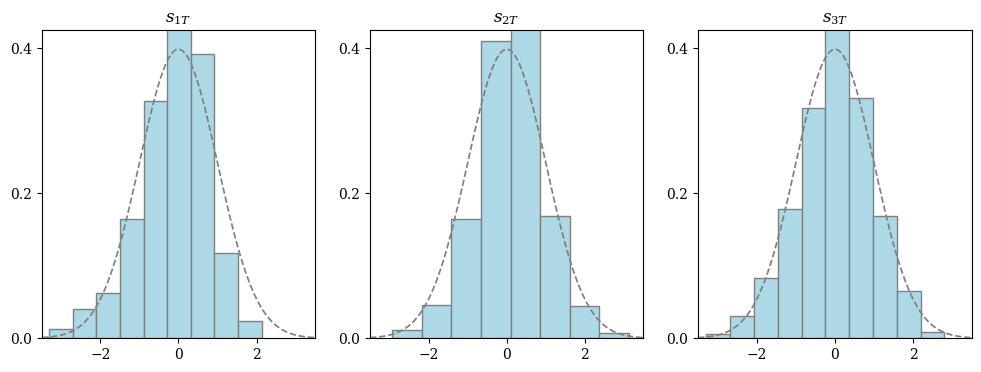}
    \caption{Histograms of $s_{iT}$ for $i=1,2,3$. $(T=100)$}
    \label{fig:hist1}
\end{figure}

\begin{figure}[H]
    \centering
    \includegraphics[scale=0.5]{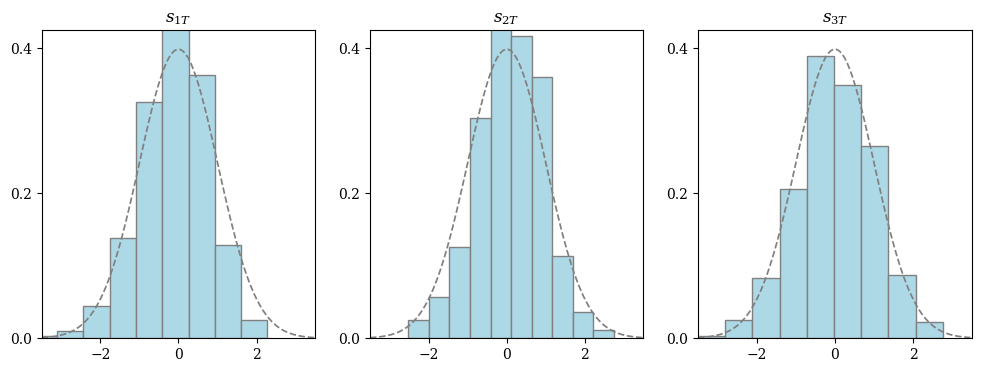}
    \caption{Histograms of $s_{iT}$ for $i=1,2,3$. $(T=200)$}
    \label{fig:hist2}
\end{figure}

\begin{figure}[H]
    \centering
    \includegraphics[scale=0.5]{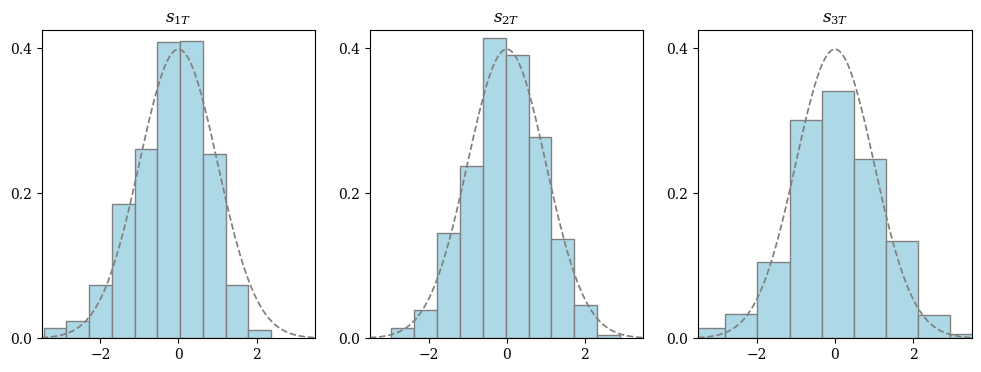}
    \caption{Histograms of $s_{iT}$ for $i=1,2,3$. $(T=500)$}
    \label{fig:hist3}
\end{figure}
\begin{table}[H]
  \caption{Empirical tail probabilities}
  \label{table:size1}
  \centering
  \begin{tabular}{lccc}
    \hline
      & $|s_{1T}|>1.96$ & $|s_{2T}|>1.96$ & $|s_{3T}|>1.96$ \\
    \hline \hline
    $T=100$  & 0.100  & 0.048 & 0.030  \\
    $T=200$  & 0.056 & 0.036 & 0.027  \\
    $T=500$  & 0.061 & 0.047 & 0.039  \\
    \hline
  \end{tabular}
\end{table}

\begin{figure}[H]
    \centering
    \includegraphics[scale=0.5]{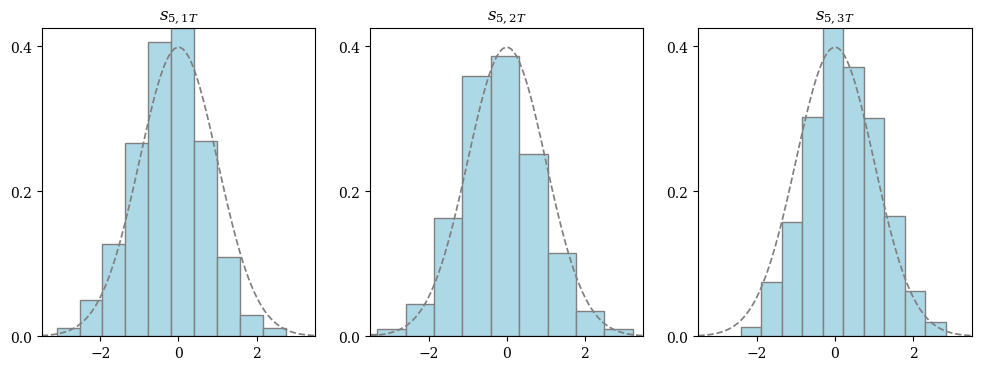}
    \caption{Histograms of $s_{5, h T}$ for $h=1,2,3$. $(T=100)$}
    \label{fig:hist4}
\end{figure}

\begin{figure}[H]
    \centering
    \includegraphics[scale=0.5]{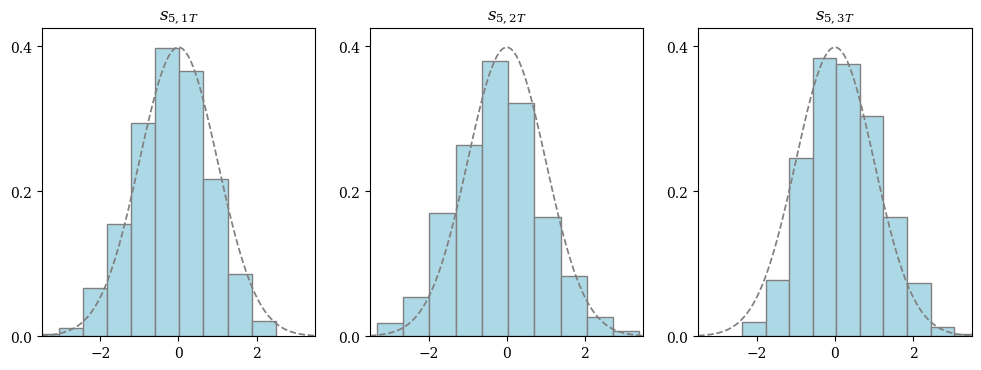}
    \caption{Histograms of $s_{5, h T}$ for $h=1,2,3$. $(T=200)$}
    \label{fig:hist5}
\end{figure}

\begin{figure}[H]
    \centering
    \includegraphics[scale=0.5]{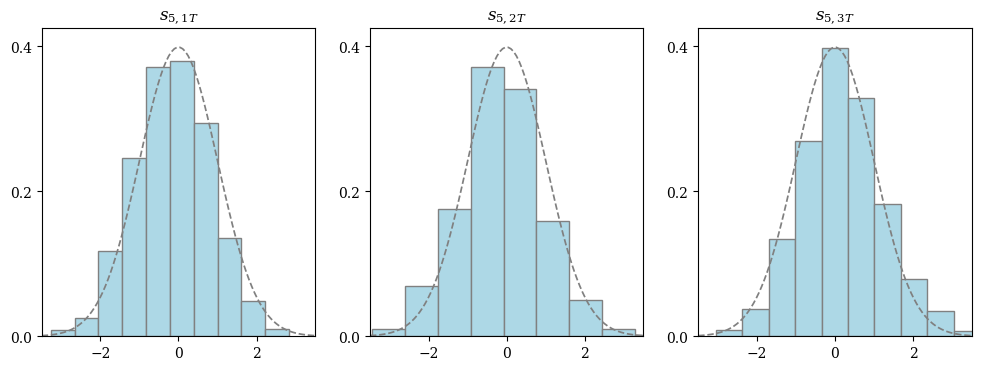}
    \caption{Histograms of $s_{5, h T}$ for $h=1,2,3$. $(T=500)$}
    \label{fig:hist6}
\end{figure}

\begin{table}[H]
  \caption{Empirical tail probabilities}
  \label{table:size2}
  \centering
  \begin{tabular}{lcccc}
    \hline
      & $|s_{5,1T}|>1.96$ & $|s_{5,2T}|>1.96$ & $|s_{5,3T}|>1.96$ \\
    \hline \hline
    $T=100$  & 0.048  & 0.056 & 0.035 \\
    $T=200$  & 0.051 & 0.075 & 0.050 \\
    $T=500$  & 0.044 & 0.078 & 0.069 \\
    \hline
  \end{tabular}
\end{table}

Finally, we check the behaviors of the test
statistics $z_{iT}, i=1,2,3$ with $\bm{v} = \bm{1}=(1,\ldots,1)^\top \in \mathbb{R}^{k(k-1)/2}$
under the alternative hypothesis 
$\mathcal{H}_1$.
Figures \ref{fig:hist alt1}--\ref{fig:hist alt3} show  
the histograms of $z_{iT}, i=1,2,3$
under $\mathcal{H}_1$ and 
the dashed line indicates the density function 
of the standard normal distribution.
Moreover, the powers of the test statistics
are summarized in Table \ref{table:power}.
The consistency of the test also seems valid.
In terms of the power, 
it may be better to use $z_{3T}$ than $z_{1T}$ and 
$z_{2T}$ for the test.
\begin{table}[H]
  \caption{Power of the test statistics}
  \label{table:power}
  \centering
  \begin{tabular}{lccc}
    \hline
      & $z_{1T}$ & $z_{2T}$ & $z_{3T}$ \\
    \hline \hline
    $T=100$  & 0.350  & 0.225 & 0.668 \\
    $T=200$  & 0.554 & 0.459 & 0.925 \\
    $T=500$  & 0.877 & 0.846 & 1.000 \\
    \hline
  \end{tabular}
\end{table}

\begin{figure}[H]
    \centering
    \includegraphics[scale=0.5]{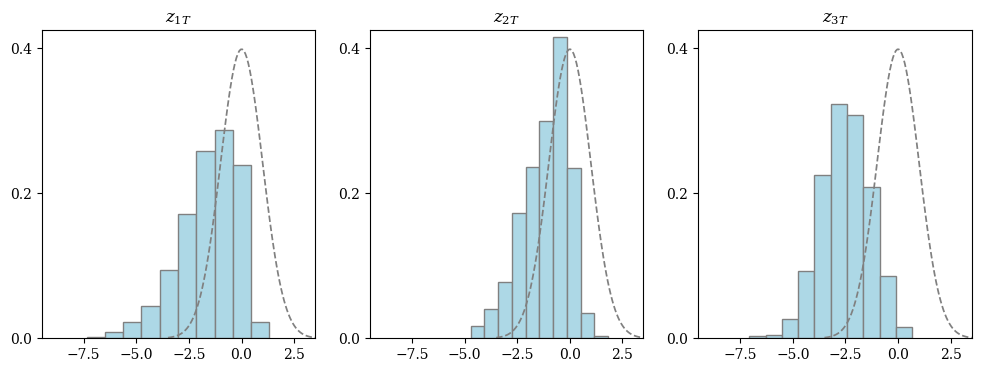}
    \caption{Distribution of test statistics under $\mathcal{H}_1$ $(T=100)$}
    \label{fig:hist alt1}
\end{figure}

\begin{figure}[H]
    \centering
    \includegraphics[scale=0.5]{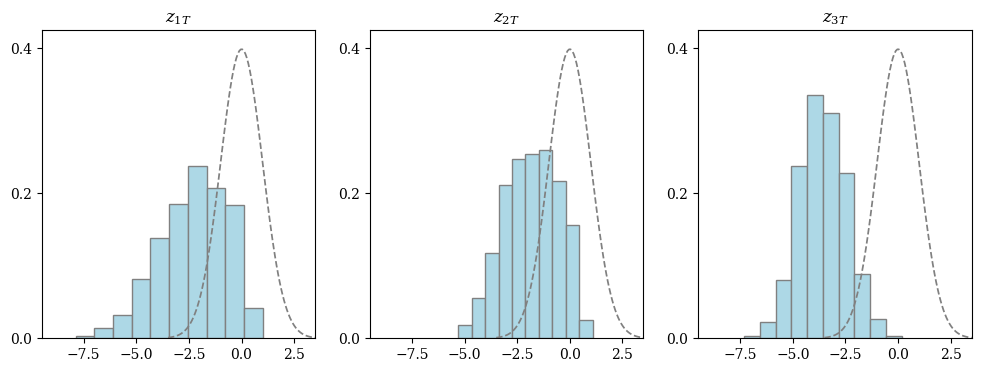}
    \caption{Distribution of the test statistics under $\mathcal{H}_1$ $(T=200)$}
    \label{fig:hist alt2}
\end{figure}

\begin{figure}[H]
    \centering
    \includegraphics[scale=0.5]{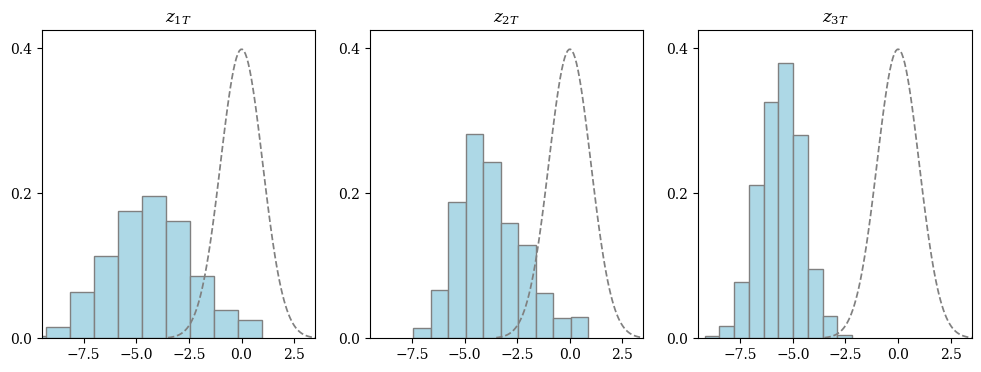}
    \caption{Distribution of the test statistics under $\mathcal{H}_1$ $(T=500)$}
    \label{fig:hist alt3}
\end{figure}
\subsection{Empirical studies}\label{subsec:real data}
We analyze the impact of policy rates on unemployment 
and prices in U.S. from 1992 to 2018, which are 
quarterly data.
We use the seasonally adjusted 
unemployment level data (U.S. Bureau of Labor Statistics, Unemployment Level [UNEMPLOY], retrieved from FRED, Federal Reserve Bank of St. Louis; https://fred.stlouisfed.org/series/UNEMPLOY, February 18, 2025), 
the seasonally adjusted core consumer price index data (U.S. Bureau of Labor Statistics, Consumer Price Index for All Urban Consumers: All Items Less Food and Energy in U.S. City Average [CPILFESL], retrieved from FRED, Federal Reserve Bank of St. Louis; https://fred.stlouisfed.org/series/CPILFESL, February 18, 2025), 
and the federal funds effective rate (Board of Governors of the Federal Reserve System (US), Federal Funds Effective Rate [DFF], retrieved from FRED, Federal Reserve Bank of St. Louis; https://fred.stlouisfed.org/series/DFF, 
February 18, 2025).
In the sequel, we modify the data and 
use the following notations.
\begin{itemize}
    \item $Y_{1,t}$ : growth rate of the unemployment level (\%).
    \item $Y_{2,t}$ : inflation rate (growth rate of the core consumer price index) (\%).
    \item $Y_{3,t}$ : difference of federal funds effective rate (\%).
\end{itemize}
Figure \ref{fig:real data} shows the plot of 
the time series $\{Y_{i, t}\}, i=1,2,3$.
\begin{figure}[H]
    \centering
    \includegraphics[scale=0.5]{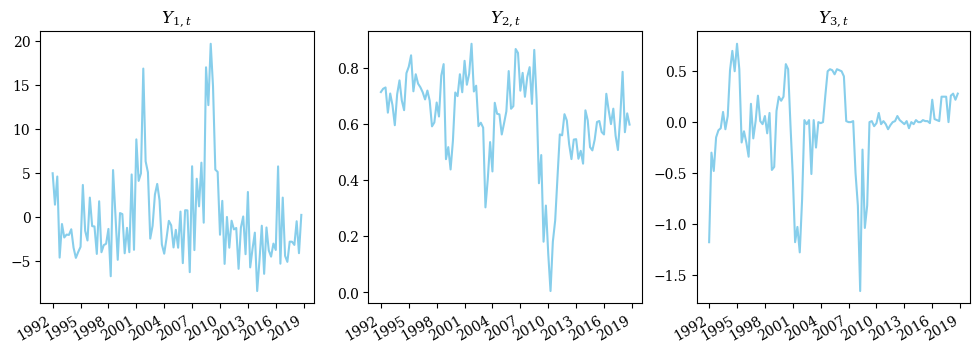}
    \caption{Plot of time series $Y_{i, t}, i=1,2,3$}
    \label{fig:real data}
\end{figure}
We apply the proposed method to the modified data.
We choose the lag $p=4$, which 
enables us to consider the correlations
between the three time series 
up to one year ago.
The estimated coefficient matrices 
of the reduced form are as follows (the number of observations is 
$T = 104$).
\begin{align*}
    \hat{\bm{B}}_{1T} &=
    \begin{pmatrix}
        0.061 & 0.634 & -3.973 \\
        -0.003 &  0.549 & -0.037 \\
        -0.017 & -0.37 & 0.651
    \end{pmatrix}, \hspace{6pt} 
    \hat{\bm{B}}_{2T} =
    \begin{pmatrix}
        0.368 & 4.763 & -1.284 \\
        -0.037 & -0.003 &  0.115 \\
        0.651 &  0.005 & 0.258
    \end{pmatrix} \\
    \hat{\bm{B}}_{3T} &=
    \begin{pmatrix}
        0.119 & 1.575 & 0.780  \\
        -0.001 &  0.201 &  0.054 \\
        -0.006 &  0.235 & -0.026
    \end{pmatrix}, \hspace{6pt} 
    \hat{\bm{B}}_{4T} =
    \begin{pmatrix}
        -0.04 & -1.956 & 1.127 \\
        -0.003 & -0.109 & -0.026\\
        0.003 & -0.311 & -0.159
    \end{pmatrix}
\end{align*}
Figure \ref{fig:total effect real data} shows the 
estimated curve of the impulse responses
and their $0.95$ confidence intervals.

\begin{figure}[H]
    \centering
    \includegraphics[scale=0.6]{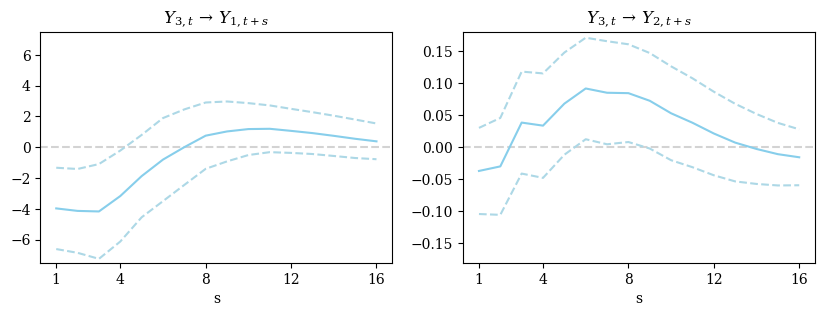}
    \caption{non-orthogonalized impulse response function}
    \label{fig:total effect real data}
\end{figure}
Note that lag $s=4$ means one year.
Since we assume that $Y_{3,t}$ does not affect $Y_{1,t}$ and $Y_{2,t}$ (this assumption derives from the ordering of the components of $\bm{Y}_{t}$), 
non-orthogonalized impulse responses in Figure \ref{fig:total effect real data} can be regarded as 'total effects' from Proposition \ref{prop:impulse response} in Appendix.
We expect that a policy rate decrease 
will cause an improvement in  
unemployment.
The upper curve of the 
significance interval in the left figure seems consistent 
with our intuition.
Meanwhile, the lower curve in the right figure 
seems to well describe the 
impact of the policy rate on the prices 
because the prices may be depressed
if the policy rate increases.
However, 
given that the monetary policy from 1992
to 2018 had often been shifted about every three years, the long-term effects shown by estimated values are more reasonable than the edges of significance intervals.
We checked the accuracy of the estimated model 
in terms of the standard deviation, root mean squared error, and adjusted R-squared coefficient,  which are summarized in  
Table \ref{table:prediction real data}.
The results indicate that 
the model is applicable to fit the data.
\begin{table}[H]
  \caption{Accuracy of the estimated model}
  \label{table:prediction real data}
  \centering
  \begin{tabular}{lccc}
    \hline
     & $\{\hat{Y}_{1,t}\}$ & $\{\hat{Y}_{2,t}\}$ & $\{\hat{Y}_{3,t}\}$  \\
    \hline \hline 
       Standard deviation & 3.557 & 0.130 & 0.306 \\
       Root mean squared error & 3.631 & 0.092 & 0.260 \\
       Adjusted R-squared & 0.416 & 0.611 & 0.552 \\
    \hline
  \end{tabular}
\end{table}

To estimate the coefficient matrices $\bm{A}_s, s=0,1,\ldots,4$ of the SVAR model, 
we consider the following assumptions
corresponding to 
Assumption \ref{assump:regularity A}.
\begin{assumption}\label{assump:real data}
\begin{itemize}
    \item[(1)] 
    The difference of the policy rate 
    $\{Y_{3, t}\}_{t}$ does not 
    depend on $(Y_{1, t-4}, Y_{2, t-4})^\top$ directly. 
    \item[(2)] 
    $Y_{3, t}$ does not affect $Y_{2, t+1}$ directly.
    \item[(3)] 
    $Y_{3, t}$ affects $Y_{3, t+1}$ directly.
    \item[(4)] 
    $Y_{1, t}$ affects $Y_{2, t+4}$ directly.
    \item[(5)]
    $Y_{3, t}$ affects $Y_{1,t+4}$ directly.
\end{itemize}
\end{assumption}
\begin{rem}
Each assumption is considered based on the 
following background.
\begin{itemize}
\item[(1)]
The policy rates are determined based on the 
behavior of the current economic indicators 
by the Federal Open Market Committee per 
about six weeks.
\item[(2)]
Prices of commodities are considered to be 
directly affected 
by production 
costs including wages of workers.
In this study, we simply assume that 
the policy rates affect prices through unemployment level.
\item[(3)]
The policy rate tends to be gradually changing.
Here, we attempt to approximate the causal structure by simplifying $Y_{3,t}$ as a direct cause of $Y_{3,t+1}$. 
\item[(4)]
The condition is a standard assumption 
because the relationship between inflation and  unemployment is widely observed.
\item[(5)]
Naturally, 
the policy rate is considered to affect
unemployment through 
various economic activities.
In this study, we simply assume that 
$Y_{3, t}$ is a direct cause of $Y_{1, t+4}$.
\end{itemize}
\end{rem}
Under Assumption \ref{assump:real data},
we consider the following 
restrictions for $\bm{A}_s, 0 \leq s \leq 4$.

\begin{align*}
    \bm{A}_0 &=
    \begin{pmatrix}
    0 & 0 & 0 \\
    * & 0 & 0 \\
    * & * & 0
    \end{pmatrix}, \hspace{6pt}
    \bm{A}_1 =
    \begin{pmatrix}
    * & * & * \\
    * & * & * \\
    * & * & \dagger
    \end{pmatrix}, \hspace{6pt}
    \bm{A}_4 =
    \begin{pmatrix}
    * & * & \dagger \\
    \dagger & * & 0 \\
    0 & 0 & 0
    \end{pmatrix},
\end{align*}
where 
$\dagger$ and $*$ mean 
nonzero parameters and arbitrary 
$\mathbb{R}$-valued parameters, respectively.
Then, the estimated coefficient 
matrices of the SVAR model 
are as follows.
\begin{align*}
    \hat{\bm{A}}_{0T} &=
    \begin{pmatrix}
        0.000 & 0.000 & 0.000 \\
        -0.023 & 0.000 & 0.000 \\
        -0.126 & 0.648 & 0.000 
    \end{pmatrix}, \hspace{6pt} 
    \hat{\bm{Q}}_{T} =
    \begin{pmatrix}
        1.000 & 0.000 & 0.000 \\
        -0.023 & 1.000 & 0.000 \\
        -0.141 & 0.648 & 1.000 
    \end{pmatrix} \\
    \hat{\bm{A}}_{1T} &=
    \begin{pmatrix}
        0.061 & 0.634 & -3.973 \\
        -0.002 & 0.564 & -0.130 \\
        -0.008 & -0.646 & 0.176
    \end{pmatrix}, \hspace{6pt} 
    \hat{\bm{A}}_{2T} =
    \begin{pmatrix}
        0.368 & 4.763 & -1.284 \\
        0.006 & 0.226 & -0.029 \\
        0.053 & 0.781 & -0.113
    \end{pmatrix} \\
    \hat{\bm{A}}_{3T} &=
    \begin{pmatrix}
        0.119 & 1.575 & 0.780 \\
        0.002 & 0.238 & 0.073 \\
        0.009 & 0.302 & 0.037
    \end{pmatrix}, \hspace{6pt} 
    \hat{\bm{A}}_{4T} =
    \begin{pmatrix}
        -0.040 & -1.956 & 1.127 \\
        -0.004 & -0.155 & 0.000 \\
        0.000 & -0.486 & 0.000
    \end{pmatrix}.
\end{align*}
Moreover, we consider the 
following test 
\[
\mathcal{H}_0: \bm{A}_0 = \bm{O},\quad
\mathcal{H}_1: \bm{A}_0 \neq \bm{O}.
\]
Based on the test statistics $z_{3T}$ defined 
in the previous section with 
$\bm{v} = \bm{1} = (1,\ldots,1)^\top \in \mathbb{R}^{k(k-1)/2}$,
the results are summarized in Table \ref{table:test real data}.
\begin{table}[H]
\caption{Test result}
  \label{table:test real data}
  \centering
  \begin{tabular}{cccc}
    \hline
     Test statistics & $T$ & $z$-value & $p$-value (two-sided test) \\
    \hline \hline
    $z_{3T}$ & 104 & -1.79 & 0.0734 \\
    \hline
  \end{tabular}
\end{table}
We cannot reject the null hypothesis 
at a significance level of $\alpha=0.05$.
As described in Proposition \ref{prop:impulse response} in Appendix, 
when $\bm{A}_0 = \bm{O}$, 
all non-orthogonalized impulse responses 
can be regarded  
not only  as partial effects, but also as ``total effects".
Therefore, if the null hypothesis
is accepted, 
we can treat non-orthogonalized impulse responses 
from not only $e_{k,t}$ but also $e_{j,t}, j=1,...,k-1$ 
as indicators of the overall effects.
However, 
given that 
we obtained a relatively small 
$p$-value for a small observations $T=104$, 
we continue our discussion 
on the estimated value of the SVAR model.
Because one of the purposes of the policy rate is 
to stabilize employment and prices,
the signs of the third row 
of $\hat{\bm{A}}_{0T}$ may be consistent 
with our intuition.
The sign of the $(2,1)$ component 
of $\hat{\bm{A}}_{0T}$ also 
seems reasonable because 
unemployment may cause a decline in 
wages, and hence, 
have a negative impact on prices. 
However, the absolute value of the 
$(2, 1)$ component
may be overestimated because 
prices are rigid.
Although we simplified the causal structure, obtaining an $\hat{\bm{A}}_{0T}$ that aligns with intuition may serve as evidence that the model can approximate the structure to some extent.

In summary, the proposed model 
seems to describe the data structure.
In particular, the non-orthogonalized 
impulse responses suggest that 
it may take about one year for 
the policy rate decrease to 
improve unemployment
and policy rates increase 
may depress the prices in the short term.
See the Appendix for the causal interpretation 
of the non-orthogonalized impulse responses.
\section{Proofs}\label{sec:proofs}
\begin{proof}[Proof of Theorem \ref{thm:normality A}]
It follows from the basic properties of the 
$\ve$ operator and the 
matrix Kronecker product that 
\begin{eqnarray*}
\bm{q}_*
&=& \ve(\bm{Q}_*) 
= \ve(\bm{L}_g(\bm{B}_*)) = f_1(\bm{b}_*) \\
\bm{a}_{0*}
&=& \ve(\bm{A}_{0*}) = \ve(\bm{I}_k - \bm{L}_g(\bm{B}_*)^{-1})
= f_2(\bm{b}_*) \\
\bm{a}_*
&=& \ve(\bm{Q}_*^{-1}\bm{B}_*) 
= \left(\bm{I}_r \otimes \bm{L}_g(\bm{B}_*)
\right)^{-1} \bm{b}_*
= f_3(\bm{b}_*).
\end{eqnarray*} 
We show that the map $\bm{L}$ is differentiable at $g(\bm{B}_*) = g(\ve^{-1}(\bm{b}_*))$.
For every matrix $\bm{C}=(c_{ij})_{1 \leq i, j \leq k}$ that admits the 
LU decomposition, the corresponding
lower-triangular matrix
$\bm{L}(\bm{C})=(l_{ij})_{1 \leq i, j \leq k}$ and 
the upper-triangular matrix
$\bm{U}(\bm{C}) = (u_{ij})_{1 \leq i, j \leq k}$
are solutions to the 
following equation:
\[
\begin{pmatrix}
c_{11} & c_{12} & c_{13} & \cdots & c_{1k} \\
c_{21} & c_{22} & c_{23} & \cdots & c_{2k} \\
\vdots & \vdots & \vdots & \ddots & \vdots \\
c_{k1} & c_{k2} & c_{k3} & \cdots &  c_{kk}
\end{pmatrix}
= \begin{pmatrix}
        1 & 0 & 0 & \cdots & 0 \\
        l_{21} & 1 & 0 & \cdots & 0 \\
        l_{31} & l_{32} & 1 & \cdots & 0 \\
        \vdots & \vdots & \vdots & \ddots & \vdots \\
        l_{k1} & l_{k2} & l_{k3} & \cdots & 1 \\
    \end{pmatrix}
\begin{pmatrix}
        u_{11} & u_{12} & u_{13} & \cdots & u_{1k} \\
        0 & u_{22} & u_{23} & \cdots & u_{2k} \\
        0 & 0 & u_{33} & \cdots & u_{3k} \\
        \vdots & \vdots & \vdots & \ddots & \vdots \\
        0 & 0 & 0 & \cdots & u_{kk} \\
    \end{pmatrix}.
\]
We can observe that for every $i, j =1,2,\ldots,k$,
$l_{ij}$ and $u_{ij}$ can be 
written as 
a rational function 
of some components of $\bm{C}$.
Under our assumptions, 
$\bm{Q}_*$ and $g(\bm{A}_*)$ are non-singular
lower and upper-triangular 
matrices, respectively, 
which implies that 
all leading principal sub-matrices 
are non-singular, therefore, 
$g(\bm{B}_*) = \bm{Q}_* g(\bm{A}_*)$ 
admits 
the LU decomposition.
Combining these facts, we obtain 
the differentiability 
of $\bm{L}$ at $g(\bm{B}_*)$.
Thus, we can apply the delta method to 
derive the conclusion.
\end{proof}
We omit the proof of 
Propositions  
\ref{prop:estimator psi}, and \ref{prop:estimator total effect} because they are direct consequences of 
the delta method.
\begin{proof}[Proof of Theorem \ref{thm:test}]
We provide the proof of the assertion (ii),
because (i) is a direct consequence of 
Lemma \ref{lem:estimation sub-vector}.

We consider the case where $l=1$.
Note that
\[
\{|z_{1T}| > c\}
=
\left\{
\left|\frac{z_{1T}}{\sqrt{T}}\right|
> \frac{c}{\sqrt{T}}
\right\}
= \{
A_T^+ > 0
\} \cup 
\{
A_T^- <0
\},
\]
where 
\[
A_T^+ = \frac{z_{1T}}{\sqrt{T}}-\frac{c}{\sqrt{T}}\quad
\mbox{and}\quad
A_T^- = \frac{z_{1T}}{\sqrt{T}}+\frac{c}{\sqrt{T}}.
\]
Noting that $\hat{\bm{\Sigma}}_1$ and 
$\hat{\bm{q}}_T$ are consistent 
estimators
for $\bm{\Sigma}_1$ and $\bm{q}_*$, respectively, 
we have 
\begin{eqnarray*}
\frac{z_{1T}}{\sqrt{T}}
&=& (\bm{v}^\top \hat{\bm{\Sigma}}_{1\mathrm{sub}} \bm{v})^{-1/2} \bm{v}^\top (\hat{\bm{q}}_{\mathrm{sub}} -\bm{q}_{*\mathrm{sub}})
+(\bm{v}^\top \hat{\bm{\Sigma}}_{1\mathrm{sub}} \bm{v})^{-1/2} \bm{v}^\top \bm{q}_{*\mathrm{sub}} \\
&=& o_p(1) + (\bm{v}^\top {\bm{\Sigma}}_{1\mathrm{sub}} \bm{v})^{-1/2} \bm{v}^\top \bm{q}_{*\mathrm{sub}},\quad
T \to \infty.
\end{eqnarray*}
Therefore, it holds that 
\[
A_T^+ = A_T^- = o_p(1) + A,
\quad T \to \infty,
\]
with 
\[
A=(\bm{v}^\top {\bm{\Sigma}}_{1\mathrm{sub}} \bm{v})^{-1/2} \bm{v}^\top \bm{q}_{*\mathrm{sub}}.
\]
By the continuous mapping theorem, we have
$1_{\{A_T^+ >0\}}
\to^p 1_{\{A >0\}}$ and 
$1_{\{A_T^- <0\}}
\to^p 1_{\{A <0\}}$ as $T \to \infty$.
Therefore, it follows from the 
dominated convergence theorem that 
\begin{eqnarray*}
\lim_{T \to \infty}\Prob(|z_{1T}| > c)
&=& \lim_{T \to \infty} \left\{
\Prob (A_T^+>0) + \Prob (A_T^-<0)
\right\} \\
&=&  \lim_{T \to \infty} \left\{
\E \left[1_{\{A_T^+>0\}}\right] + \E\left[ 1_{\{A_T^-<0\}} \right]
\right\} \\
&=& \E\left[1_{\{A^+ >0\}}\right] + \E\left[1_{\{A^- <0\}}\right]= 1,
\end{eqnarray*}
which ends the proof for $l=1$.
The assertion for $l=2,3$ can be proved similarly.
\end{proof}
\appendix
\section{Appendix: Causal interpretations of SVAR models}\label{sec:appendix}
We discuss the causal interpretations 
of SVAR models based on 
linear structural equation models 
and their causal diagrams.
\subsection{Linear structural equation models and causal diagrams}\label{subsec:causal}
As a preliminary step, we present 
the basic properties of 
linear structural equation models.
Let $\bm{V}$ be a set of a finite number of random variables.
For every $X \in \bm{V}$, let $\bm{PA}(X)$ be the parents of 
$X$ defined as follows:
\[
\bm{PA}(X) = \{Y \in \bm{V} : Y \mbox{ is a direct cause of }X\}.
\]
Then, a linear structural equation model is 
expressed as follows:
\begin{equation}\label{eq:SEM}
X_i = \sum_{X_j \in \bm{PA}(X_i)} a_{ij}X_j + u_i, 
\quad X_i \in \bm{V}.
\end{equation}
We call the constant $a_{ij} \neq 0$ a path coefficient
from $X_j$ to $X_i$.
The error term $u_i$ is a direct cause of $X_i$, which cannot be represented by the 
random variables 
in $\bm{V}$.
Then, we can consider the causal diagram studied by 
e.g., \cite{Wright1921} and \cite{Pearl2009}, 
$\bm{G} = (\bm{V}, \bm{E})$ for \eqref{eq:SEM}, where 
the vertex set $\bm{V}$ is a set of random variables and 
the set of edges $\bm{E}$ corresponds to the 
path coefficients for pairs of $\bm{V}$
and correlations among error terms. 
Similarly to Example \ref{eg:SVAR-ex}, we 
regard
correlations among error terms as the existence of common causes.
In the sequel, we introduce some definitions 
associated with linear structural equation models.
\begin{defi}\label{def:effects}
Let $\bm{G} = (\bm{V}, \bm{E})$ be a causal 
diagram for a linear structural equation model 
\eqref{eq:SEM}.
\begin{itemize}
\item[(i)]
The direct effect from $X_j$ to $X_i$ is the path coefficient $a_{ij}$, which is a directed edge from 
$X_j$ to $X_i$.
If such a directed edge does not exist, 
the direct effect is set to $0$.
\item[(ii)]
The indirect effect from $X_j$ to $X_i$ is  
the sum of the 
multiplications of the 
path coefficients corresponding to the  
directed paths except for the directed edge 
from $X_j$ to $X_i$.
\item[(iii)]
The total effect $\tau_{ij}$ from $X_j$ to $X_i$ is 
the sum of the direct and indirect effects.
\end{itemize}
\end{defi} 
See \cite{Alwin1975} for the details of the causal effects
defined above.
We further define the following controlled total 
effect.
\begin{defi}\label{def:controlled effects}
Let $\bm{G} = (\bm{V}, \bm{E})$ be a causal 
diagram for the linear structural equation model 
\eqref{eq:SEM}.
For a vertex set $\bm{S} \subset \bm{V}$ such as 
$\bm{S} \cap \{X_i, X_j\} = \emptyset$,
consider the subgraph $\bm{G}'$  
constructed by removing the edges whose terminals are vertices of $\bm{S}$.
Then, the $\bm{S}$-controlled total effect 
$\tau_{ij|do(\bm{S})}$ means the total effect
from $X_j$ to $X_i$ restricted to the subgraph $\bm{G}'$.
\end{defi}
The term ``controlled'' means ``atomic intervention'' 
(see \cite{Pearl1995} for the detail).
The following lemma provides the different representations of the 
linear structural equation model \eqref{eq:SEM}.
\begin{lem}\label{lem:ancestor expansion}
For every $X_i \in \bm{V}$, let $\bm{AN}(X_i)$
be the ancestor set of $X_i$ defined as follows:
\[
\bm{AN}(X_i)
= \{Y \in \bm{V} : \mbox{There exists at least one directed path from }Y\mbox{ to }X_i\}.
\] 
For the causal diagram $\bm{G}$ associated with 
model \eqref{eq:SEM}, 
let $\tilde{\bm{G}}$ be a subgraph of $\bm{G}$
constructed by removing bidirectional 
arrows
corresponding to the correlations between 
error terms.
Suppose that $\tilde{\bm{G}}$ is a 
directed acyclic graph (DAG).
For every $\bm{S} \subset \bm{V}$ such that 
$X_i \not \in \bm{S}$, it holds that 
\begin{equation}\label{eq:ancestor expansion}
X_i = \sum_{X_j \in \bm{S}} \tau_{ij \, | \, do(\bm{S}_{-j}) } X_j + \sum_{X_j \in \bm{AN}(X_i) \setminus \bm{S}} \tau_{ij|do(\bm{S})}u_j + u_i,
\end{equation}
where $\bm{S}_{-j} = \bm{S} \setminus \{X_j\}$.
\end{lem}
\begin{proof}
Noting  that $\tilde{\bm{G}}$ is 
a DAG, we can 
assume without loss of generality, 
that for any indices $i <j$ of 
$\bm{V} = \{X_1,\ldots,X_N\}$, 
$X_j$ is not an ancestor of $X_i$.
Thus, we have 
\begin{equation}\label{eq:proof ancestor1}
X_i = \sum_{1 \leq i < j} a_{ij}X_j + v_i,
\end{equation}
where $a_{ij}$ represents the direct effect from $X_j$ to $X_i$ here, to discuss the general case. 
Substituting the model equation for the 
parent $X_l$ of $X_i$ in the summand of the right-hand side of \eqref{eq:proof ancestor1}, 
we have 
\begin{equation}\label{eq:proof ancestor2}
X_i = \sum_{1 \leq j < i, j \neq l} a_{ij} X_j 
+ a_{il} \left(\sum_{1 \leq j <l} a_{lj} X_j + v_l\right)
+ v_i.
\end{equation}
If $a_{il} a_{lj} \neq 0$, 
this coefficient corresponds to 
the directed path $X_j \to X_l \to X_i$.
Because $\bm{V}$ is finite and 
$\tilde{\bm{G}}$ is a DAG, 
sequential substitution yields the
conclusion.
\end{proof}
We call equality 
\eqref{eq:ancestor expansion} 
the ancestor expansion of $X_i$.
Then, we summarize the validity of the total effect and 
the controlled total effect as indicators of causal relationships.
It follows from Lemma \ref{lem:ancestor expansion}
with $\bm{S} = \{X_j\}$ that 
\[
X_i = \tau_{ij} X_j + \sum_{X_j \in \bm{AN}(X_i) \setminus \{X_j\}} \tau_{ij|do(\{X_j\})}u_j + u_i.
\]
Since 
we regard correlations among error terms as the existence of common causes, $X_j$ cannot affect $X_i$ through error terms.
Therefore, 
$\tau_{ij}$ can be interpreted as the
effect of $X_i$ caused by $X_j$.
More generally, we have 
\[
X_i = \sum_{X_j \in \bm{S}} \tau_{ij| do(\bm{S}_{-j}) } X_j + \sum_{X_j \in \bm{AN}(X_i) \setminus \bm{S}} \tau_{ij |do(\bm{S})}u_j + u_i,
\]
which implies that 
$\tau_{ij| do(\bm{S}_{-j}) }$ is the effect for 
$X_i$ caused by $X_j$s 
under the atomic intervention of all variables in 
$\bm{S}_{-j}$.
In other words, $\tau_{ij | do(\bm{S}_{-j})}$ means the sum of effects from $X_j$ to $X_i$ that does not go through components of $\bm{S}_{-j}$.
See \cite{Pearl1998} for the details of such 
interpretations.

We provide a simple example.
\begin{eg}\label{eg:SEM-1}
Let $\bm{V} = \{X_1, X_2, X_3, X_4\}$, 
and $u_1, u_2, u_3, u_4$ be possibly correlated error terms.
Then, we consider the following structural equation model:
\begin{equation}\label{eq:SEM-1}
\begin{cases}
    X_1 &= u_1 \\
    X_2 &= a_{21}X_1 + u_2 \\
    X_3 &= a_{31}X_1 + a_{32}X_2 + u_3 \\
    X_4 &= a_{42}X_2 + a_{43}X_3 + u_4.
\end{cases}
\end{equation}
The causal diagram for \eqref{eq:SEM-1}
is shown in Figure \ref{fig:SEM-1}.
\begin{figure}[H]
    \centering
    \includegraphics[scale=0.5]{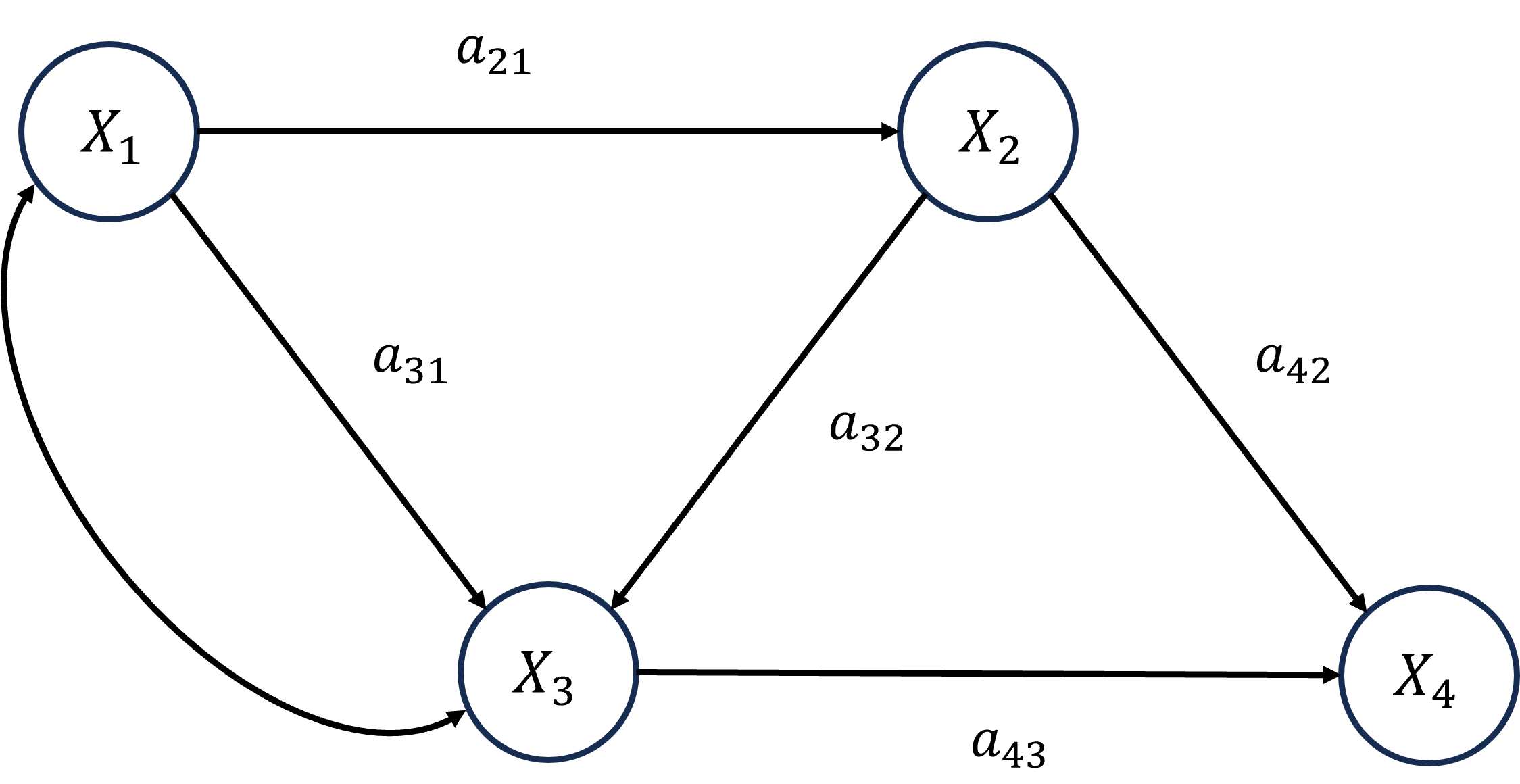}
    \caption{Causal diagram for \eqref{eq:SEM-1}}
    \label{fig:SEM-1}
\end{figure}
The bidirectional arrow indicates that 
$u_1$ and $u_3$ are correlated.
Note that this diagram is a DAG 
when we ignore the bidirectional arrow or, 
when $u_1,\ldots,u_4$
are mutually uncorrelated.
For instance, we have 
\begin{eqnarray*}
\tau_{31} &=& a_{31} + a_{32}a_{21}, \\
\tau_{41} &=& a_{42}a_{21} + a_{43}a_{32}a_{21} + a_{43}a_{31}, \\
\tau_{41\,|\,do(\{X_3\}) } &=& a_{42}a_{21}.
\end{eqnarray*}
It follows from Lemma \ref{lem:ancestor expansion} that 
\begin{eqnarray*}
    X_4 &=& a_{42}X_2 + a_{43}X_3 + u_4 \\
    &=& a_{42}(a_{21}X_1 + u_2) + a_{43}X_3 + u_4 \\
    &=& a_{42}a_{21}X_1 + a_{43}X_3 + a_{42}u_2 + u_4 \\
    &=& \tau_{42 \, | \, do(\bm{S}_{-1})}X_1 + \tau_{42 |  do(\bm{S}_{-3})}X_3 + \tau_{42 | do(\bm{S})}u_2 + u_4,
\end{eqnarray*}
where $\bm{S} = \{X_1, X_3\}$.
\end{eg}
\subsection{Causal interpretations of SVAR models}\label{subsec:causal SVAR}
Notably, we have 
\begin{equation}\label{eq:SVAR-reduced componentwise}
Y_{i,t} = \sum_{j<i}^{k}a^{(0)}_{ij}Y_{j,t} + \sum_{s=1}^{p}\sum_{j=1}^{k}a^{(s)}_{ij}Y_{j,t-s} + v_{i,t},\quad
t \in \mathbb{Z},
\end{equation}
where for $s=0,\ldots,p$, 
$a^{(s)}_{ij}$ is the $(i,j)$-th component 
of $\bm{A}_s$.
This form \eqref{eq:SVAR-reduced componentwise} is the special case of the linear structural equation models 
\eqref{eq:SEM} and 
the coefficients $a^{(s)}_{ij}, s=0,\ldots,p, i, j =1,\ldots,k$ can be regarded as direct effects.
The existence of the intercept $\bm{\mu}$ does not change the following interpretations, 
thus we omit it here for simplicity.
Because the ``error terms'' of the linear structural 
equation model mean the direct causes 
that cannot be represented by variables in $\bm{V}$,
they depend on the choice of $\bm{V}$.
Therefore, they are not necessarily 
equivalent to the innovation processes 
of  SVAR models.
For example, 
if we consider the case where 
$\bm{V} = \{Y_{i, t} : i=1,\ldots,k\}$ for fixed $t$, 
then the error term is 
\[
\sum_{j=1}^{k}a^{(s)}_{ij}Y_{j,t-s} + v_{i,t},
\]
which is not equal to 
the innovation process $v_{i, t}$.
However, in the sequel, we consider a 
sufficiently large $\bm{V}$
so that a finite number of error terms of $\{Y_{i,t}\} \subset \bm{V}$ that play a key role coincide with the innovation process of the SVAR model.

Noticing that the reduced form 
\[
\bm{Y}_t = \sum_{s=1}^{p} \bm{Q} \bm{A}_{s} \bm{Y}_{t-s} + \bm{Q}\bm{v}_t,\quad
\bm{Q} = (\bm{I}_k - \bm{A}_0)^{-1}
\]
corresponds to the ancestor expansion
of \eqref{eq:SVAR} for 
\[
\bm{S} = \{Y_{j, t-s} : j=1,\ldots,k,s=1,\ldots,p\},
\]
we have the following proposition.
\begin{prop}\label{prop:ancestor SVAR}
Suppose that 
Assumption \ref{assump:stationary}-(i) holds.
\begin{itemize}
\item[(i)]
For model \eqref{eq:SVAR-reduced},
it holds that 
the $(i, j)$-component of $\bm{Q}\bm{A}_s$ 
coincides with $\bm{S}_{-(j, s)}$-controlled total effect
from $Y_{j, t-s}$ to $Y_{i,t}$, where 
\[
\bm{S}_{-(j, s)}
= \{Y_{q, t-l} : q=1,2,\ldots,k,l=1,\ldots,p\}\setminus \{Y_{j, t-s}\}.
\]
\item[(ii)]
For model \eqref{eq:SVAR-reduced} and $i \neq j$,
the total effect from $Y_{j, t}$ to $Y_{i, t}$
coincides with 
the $(i, j)$-th component of $\bm{Q}$.
\end{itemize}
\end{prop}
\begin{proof}
\begin{itemize}
\item[(i)]
Because $\bm{I}_k - \bm{A}_0$ is non-singular, 
we obtain the reduced form \eqref{eq:SVAR-reduced} from \eqref{eq:SVAR}.
We obtain the 
conclusion by applying Lemma \ref{lem:ancestor expansion} for 
$\bm{S} = \{Y_{j, t-s} : j=1,2,\ldots,k, s=1,2,\ldots,p\}$.
\item[(ii)]
For $\bm{S} = \{Y_{j, t-s} : j=1,2,\ldots,k, s=1,2,\ldots,p\}$, 
it holds that
$\tau_{ij | do(\bm{S})}^{(0)} = \tau_{ij}^{(0)}$
, where the subscript (0) in the upper right represents the time difference between the cause $Y_{j,t-0}$ and the effect $Y_{i,t}$.
Because the matrix $\bm{Q}$ is the coefficient 
of $\bm{v}_t$ and the reduced form \eqref{eq:SVAR-reduced}
coincides with the ancestor expansion of $Y_{i, t}$
for 
$\bm{S}$, we have the conclusion.
\end{itemize}
\end{proof}
\begin{cor}\label{cor:ancestor SVAR}
Let $\{\bm{Y}_t\}_{t \in \mathbb{Z}}$ be 
a time series satisfying \eqref{eq:SVAR-reduced}.
For every permutation matrix 
$\bm{P} \in \mathbb{R}^{k \times k}$, 
define $\tilde{\bm{Y}}_t = \bm{P} \bm{Y}_t$.
Then, the $(i, j)$ component of $\bm{P} \bm{Q} \bm{A}_s \bm{P}^\top$
coincides with the $\tilde{\bm{S}}_{-(j, s)}$-controlled
total effect from $\tilde{\bm{Y}}_{t-s}$ to $\tilde{\bm{Y}}_t$, 
where 
\[
\tilde{\bm{S}}_{-(j, s)}
= \{
\tilde{Y}_{l, t-s} : 
l=1,2,\ldots,k, s = 1,2,\ldots,p
\} \setminus \{\tilde{Y}_{j, t-s}\}.
\]
In addition, the $(i, j)$ component of $\bm{P} \bm{Q} \bm{P}^\top$ for every $i \neq j$ coincides with the total effect
from $\tilde{Y}_{j, t}$ to $\tilde{Y}_{i, t}$.
\end{cor}
\begin{proof}
Noting that $\bm{P}$ is orthogonal, 
we have 
\begin{eqnarray*}
\bm{P} \bm{Y}_t 
&=& \sum_{s=1}^p \bm{P} \bm{Q} \bm{A}_s \bm{P}^\top \bm{P} \bm{Y}_{t-s} + \bm{P} \bm{Q} \bm{P}^\top \bm{P} \bm{v}_t \\
\tilde{\bm{Y}}_t 
&=& \sum_{s=1}^p \bm{P} \bm{Q} \bm{A}_s \bm{P}^\top \tilde{\bm{Y}}_{t-s} + \bm{P} \bm{Q} \bm{P}^\top \tilde{\bm{v}}_t,
\end{eqnarray*}
where 
$\tilde{\bm{Y}}_t = \bm{P} \bm{Y}_t$ and 
$\tilde{\bm{v}_t} = \bm{P} \bm{v}_t$.
Because $\tilde{\bm{Y}}_t$ is just a rearrangement 
of the components of $\bm{Y}_t$, 
the assertion follows from 
Proposition \ref{prop:ancestor SVAR}.
\end{proof}
Therefore, the regression coefficients of 
the SVAR model correspond to the 
causal effects in the sense of linear structural 
equation models.

Meanwhile, the causal effects for general 
time series models are described by the Granger causality 
and impulse response functions.
In the sequel, we consider the relationships
between such concepts and the causal effects in 
the linear structural equation models.
For the model \eqref{eq:SVAR-reduced componentwise}, 
the time series $\{Y_{j, t}\}_{t \in \mathbb{Z}}$
is called a Granger cause of $\{Y_{i, t}\}_{t \in \mathbb{Z}}$
if at least one coefficient 
$a_{ij}^{(s)}, s=1,\ldots,p$ is nonzero.
By using Proposition \ref{prop:ancestor SVAR},
we can interpret this Granger causality 
as the causality for the linear structural 
equation model \eqref{eq:SEM}.
\begin{prop}\label{prop:granger}
Consider the model \eqref{eq:SVAR-reduced componentwise}.
If $\{Y_{j, t}\}_{t \in \mathbb{Z}}$
is a Granger cause of $\{Y_{i, t}\}_{t \in \mathbb{Z}}$,
there exists a directed edge from $Y_{j, t-s}$ for some $s = 1,2,\ldots,p$
to $Y_{i, t}$ or 
a directed path from $Y_{j, t-s}$ to $Y_{i, t}$
depending only on the random variables at time $t$.
\end{prop}
We omit the proof because it is a 
direct application of Proposition \ref{prop:ancestor SVAR}. 
It follows from Corollary \ref{cor:ancestor SVAR} that
we can obtain the same interpretation
of the Granger causality 
regardless of the order of the components of the 
$\bm{Y}_t$.

Proposition \ref{prop:ancestor SVAR} provides the 
interpretation of the impulse responses as follows.
\begin{prop}\label{prop:impulse response}
Suppose that Assumption \ref{assump:stationary}
holds.
\begin{itemize}
\item[(i)]
The non-orthogonalized impulse response $\mathrm{IRF}_{ij}(s)$ coincides with 
the $\bm{S}_{-j}^s$-controlled total effect
from $Y_{j, t-s}$ to $Y_{i, t}$, where
\[
\bm{S}_{-j}^s = \{Y_{l, t-s} : l \neq j\}.
\]
In particular, if $\bm{A}_0 = \bm{O}$, 
any $\mathrm{IRF}_{ij}(s)$ coincides with 
the total effect from $Y_{j, t-s}$ to $Y_{i, t}$.
\item[(ii)]
Suppose that there is no 
contemporaneous
confounding 
and the innovation term $\bm{v}_t$ satisfies  
$\Var[\bm{v}_t] = \sigma_v^2 \bm{I}_k$ for 
some constant $\sigma_v^2 >0$, 
then, 
the orthogonalized impulse response 
$\mathrm{OIRF}_{ij}(s)$ coincides with 
the total effect from 
$Y_{j, t-s}$ to $Y_{i, t}$.
\end{itemize}
\end{prop}
\begin{proof}
For any $h >0$, let 
\begin{equation}\label{eq:proof V1}
\bm{V} = \{
Y_{j, t-s} : j=1,2,\ldots,k, s=0,1,\ldots,h+p
\},
\end{equation}
\begin{equation}\label{eq:proof Sh}
\bm{S}^h 
= \{
Y_{j, t-s} : j = 1,2,\ldots,k, s=h+1,\ldots,h+p
\},
\end{equation}
and 
\begin{equation}\label{eq:proof Sh minus}
\bm{S}_{-(j, s)}^h
= \bm{S}^h \setminus \{Y_{j, t-s}\}.
\end{equation}
It follows from the ancestor expansion 
of $Y_{i, t}$ for $\bm{S}^h$ that 
\begin{eqnarray}\label{eqs:ancestor impulse proof}
Y_{i, t}
&=& \sum_{Y_{j, t-s} \in \bm{S}^h} \tau_{ij|do(\bm{S}_{-(j, s)}^h)}^{(s)} Y_{j, t-s}
+ \sum_{Y_{j, t-s} \in \bm{AN}(Y_{i, t})\setminus \bm{S}^h} \tau_{ij | do(\bm{S}^h)}^{(s)} v_{j, t-s}
+ v_{i, t} \nonumber \\
&=& \sum_{s=h+1}^{h+p} \sum_{j=1}^k \tau_{ij|do(\bm{S}_{-(j, s)}^h)}^{(s)} Y_{j, t-s}
+ \sum_{s=0}^h \sum_{j=1}^k \tau^{(s)}_{ij | do(\bm{S}^h)}v_{j, t-s} \nonumber \\
&=& \sum_{s=0}^h \sum_{j=1}^k \tau_{ij}^{(s)} v_{j, t-s}
+  \sum_{s=h+1}^{h+p} \sum_{j=1}^k \tau_{ij|do(\bm{S}_{-(j, s)}^h)}^{(s)} Y_{j, t-s},
\end{eqnarray}
where 
$\tau_{ij | do(\bm{S}^h)}^{(0)} = \tau_{ij}^{(0)} = 1$
if $i=j$.
On the other hand, we have the following representation 
\begin{eqnarray}\label{eqs:expansion impulse}
\bm{Y}_t 
&=& \sum_{s=0}^h [\bm{\Lambda}^s]_k^k \bm{e}_{t-s}
+ [\bm{\Lambda}^{h+1}]_{kp}^k \bm{\xi}_{t-h-1} \nonumber \\
&=& \sum_{s=0}^h [\bm{\Lambda}^s]_k^k \bm{Q}\bm{v}_{t-s}
+ [\bm{\Lambda}^{h+1}]_{kp}^k \bm{\xi}_{t-h-1},
\end{eqnarray}
where 
for $s \geq 0$, 
$[\bm{\Lambda}^s]_k^k$ is the 
$k \times k$ sub-matrix of $\bm{\Lambda}^s$
corresponding to the top $k$ columns and rows, 
$[\bm{\Lambda}^s]_{kp}^k$ is the sub-matrix of 
$\bm{\Lambda}^s$ comprising the top $k$ rows of 
$\bm{\Lambda}^s$, 
$\bm{\xi}_t = (\bm{Y}_t^\top,\ldots,\bm{Y}_{t-p+1}^\top)^\top$,  
\[
\bm{\Lambda}
= \left(\begin{array}{ccccc}
\bm{B}_1 & \bm{B}_2 & \cdots & \bm{B}_{p-1} & \bm{B}_p \\
\bm{I}_k & \bm{O} & \cdots & \bm{O} & \bm{O} \\
\bm{O} & \bm{I}_k & \cdots & \bm{O} & \bm{O} \\
\vdots & \vdots & \ddots & \vdots & \vdots \\
\bm{O} & \bm{O} & \cdots & \bm{I}_k & \bm{O} \\
\end{array}
\right),
\]
and $\bm{B}_s = \bm{Q}\bm{A}_s$ for $s=1,\ldots,p$.
Because the two expressions 
\eqref{eqs:ancestor impulse proof} and  \eqref{eqs:expansion impulse} 
are derived from the same sequential substitutions, they are equivalent.
Therefore, the assertion (ii) is obvious because 
$\tilde{\bm{L}}=\bm{Q}$ holds from $\Var[\bm{e}_t]=\bm{Q}\sigma^2_v\bm{I}_k\bm{Q}^\top$.
Meanwhile, the coefficients of $\bm{Y}_{t-h-1}$ correspond to the left block of 
$[\bm{\Lambda}^{h+1}]^k_{kp}$, which coincides with $[\bm{\Lambda}^{h+1}]^k_k$. 
Moreover, because $Y_{j,t-h-1}$ does not have directed paths to $Y_{j,t-s}, s > h+1$, 
the $\bm{S}^{h}_{-(j,s)}$-controlled total effect from $Y_{t-h-1}$ to $Y_{i,t}$ coincides with $\bm{S}^{h+1}_{-j}$-controlled total effect, where
$\bm{S}^{h+1}_{-j}=\{Y_{l, t-h-1} : l \neq j\}$. 
Thus, the assertion (i) is proved.
\end{proof}
\begin{rem}\label{rem:discussion time series causality}
The impulse responses are often interpreted as a 
shock caused by the innovation to the future realizations 
of endogenous variables (\cite{Hamilton1994} and \cite{Sims1980}).
Meanwhile, 
Proposition \ref{prop:impulse response} 
provides the different perspectives as the 
(controlled) total effects from the endogenous variable $Y_{j,t}$ to the future realization $Y_{i,t+s}$ for some $s$. 
Particularly, it follows from Corollary \ref{cor:ancestor SVAR}, Propositions \ref{prop:granger}, and \ref{prop:impulse response} that we can regard Granger causality and non-orthogonalized impulse responses as indicators of causality even if contemporaneous confounding exists or the order of the components of the $Y_t$ is unknown.
\end{rem}
\section*{Acknowledgements}
The authors would like to
thank Enago (www.enago.jp) for the English language
review.
This work was  
supported by JSPS KAKENHI Grant Number 21K13271(K.F.).
\bibliographystyle{econ}
\bibliography{ref2.0}
\end{document}